\documentclass[10pt,a4paper,DIV=10]{scrartcl}

\usepackage[utf8]{inputenc}	%for utf8 in source
\usepackage{amsmath,amssymb,amsthm}	%standard math packages
\usepackage{mathtools}		%for \coloneqq and others
\usepackage[scr=boondoxo,scrscaled=1.05,bb=px,bbscaled=0.925]{mathalfa}
\usepackage{relsize}		%for \mathsmaller
\usepackage{csquotes}		%for easy quotation
\usepackage[outline]{contour}	%for visible text on pattern
\contourlength{1.2pt}

\usepackage{tikz}		%for the figure
\usetikzlibrary{patterns}	%ditto
\usepackage{authblk}		%multiple authors
\usepackage{siunitx}		%to typeset degree celsius
\usepackage{float}		%float positioning

\usepackage{xcolor}		%to mix colors
\definecolor{myblue}{HTML}{3D6288}
\definecolor{myred}{HTML}{9F1D2B}
\definecolor{mygreen}{HTML}{5B892D}

\usepackage{hyperref}		%cross-referencing
\hypersetup{%
	pdfencoding=auto,%
	colorlinks,%
	linkcolor=myred,%
	citecolor=mygreen,%
	urlcolor=myblue,%
}
\usepackage[capitalise]{cleveref}

\title{Probabilistic Databases with an\texorpdfstring{\\}{ }%
	Infinite Open-World Assumption}

\author[1]{Martin Grohe}
\affil[1]{\bgroup\large\url{grohe@informatik.rwth-aachen.de}\egroup}

\author[2]{Peter Lindner}
\affil[2]{\bgroup\large\url{lindner@informatik.rwth-aachen.de}\egroup}

\affil[\bgroup\empty\egroup]{\bgroup\large RWTH Aachen University\egroup}

\newcommand{\vleft}[1]
	{\mathopen{\raisebox{-4pt}{\ensuremath{\mkern4mu\Bigg#1}}}}
\newcommand{\vright}[1]
	{\mathclose{\raisebox{-4pt}{\ensuremath{\Bigg#1}}}}

\DeclareMathOperator*{\sprod}{\mathsmaller{\prod}}

\newcommand{\sfrac}[2]{\textstyle\frac{#1}{#2}}

\let\cal\mathcal
\let\fr\mathfrak
\let\scr\mathscr
\let\bf\mathbf
\let\ds\mathbb
\let\tt\mathtt

\newcommand{\cc}{\scr C}

\newcommand{\dd}{\scr D}

\newcommand{\NN}{\ds N}
\newcommand{\RR}{\ds R}
\newcommand{\UU}{\ds U}
\newcommand{\FO}{\textsf{\upshape FO}}

\let\ol\overline

\newcommand{\SubsetsOf}[1]{2^{#1}}

\renewcommand*{\setminus}{-}

\newcommand{\under}{\ensuremath{\,{\vert}\,}}
\newcommand{\bigunder}{\ensuremath{\,{\big\vert}\,}}

\DeclareMathOperator{\Ex}{E}

\newcommand{\adom}{\operatorname{adom}}
\newcommand{\ar}{\ensuremath{\operatorname{ar}}}
\newcommand{\dbinst}[1][\tau,\UU]{{\bf D[#1]}}
\newcommand{\Facts}[1][\tau,\UU]{F[#1]}

\newcommand*{\true}{\ensuremath{\textsc{True}}}
\newcommand*{\false}{\ensuremath{\textsc{False}}}
\renewcommand{\phi}{\varphi}
\newcommand*{\bid}{b.\kern1pt i.\kern1pt d.}
\newcommand*{\ti}{t.\kern1pt i.}

\renewcommand{\epsilon}{\varepsilon}
\newcommand{\EMPTY}{\ensuremath{\text{\textsc{Empty}}}}

\newtheorem{theorem}{Theorem}[section]

\newtheorem{proposition}[theorem]{Proposition}
\newtheorem{lemma}[theorem]{Lemma}
\newtheorem{corollary}[theorem]{Corollary}

\newtheorem{fact}[theorem]{Fact}
\theoremstyle{definition}
\newtheorem{example}[theorem]{Example}
\newtheorem{definition}[theorem]{Definition}
\newtheorem{remark}[theorem]{Remark}
\theoremstyle{plain}
\newtheorem*{apxclaim}{Claim}
\newenvironment{proofsketch}{\begin{proof}[Proof sketch]}{\end{proof}}

\begin{document}

% ╔═══════════════════════════════════════════════════════════════════════════╗
% ║ Title and Abstract                                                        ║
% ╚═══════════════════════════════════════════════════════════════════════════╝
\maketitle 

\begin{abstract}
Probabilistic databases (PDBs) introduce uncertainty into relational databases
by specifying probabilities for several possible instances.  Traditionally,
they are  \emph{finite} probability spaces over database instances. Such finite
PDBs inherently make a closed-world assumption: non-occurring facts are assumed
to be impossible, rather than just unlikely. As convincingly argued by Ceylan
et al. (KR '16), this results in implausibilities and clashes with intuition.
An open-world assumption, where facts not explicitly listed may have a small
positive probability can yield more reasonable results. The corresponding
open-world model of Ceylan et al., however, assumes that all entities in the
PDB come from a fixed finite universe.\par

In this work, we take one further step and propose a model of \enquote{truly}
open-world PDBs with an infinite universe. This is natural when we consider
entities from typical domains such as integers, real numbers, or strings. While
the probability space might become infinitely large, all instances of a PDB
remain finite. We provide a sound mathematical framework for infinite PDBs
generalizing the existing theory of finite PDBs. Our main results are concerned
with countable, tuple-independent PDBs; we present a generic construction
showing that such PDBs exist in the infinite and provide a characterization of
their existence. This construction can be used to give an open-world semantics
to finite PDBs. The construction can also be extended to so-called
block-independent-disjoint probabilistic databases.\par

Algorithmic questions are not the focus of this paper, but we show how query
evaluation algorithms can be lifted from finite PDBs to perform approximate
evaluation (with an arbitrarily small additive approximation error) in
countably infinite tuple-independent PDBs.
\end{abstract}

% ╒═════════╤═════════════════════════════════════════════════════════════════╕
% │ Section │ Introduction                                                    │
% ╘═════════╧═════════════════════════════════════════════════════════════════╛
\section{Introduction}
Probabilistic databases (PDBs) are uncertain databases where uncertainty is
quantified in terms of probabilities. The current standard model of
probabilistic databases
\cite{Aggarwal+2009,Green2009,Suciu+2011,VandenBroeck+2017} is an extension of
the relational model that associates probabilities to the facts appearing in a
relational database. Formally, it is convenient to view such a probabilistic
database, which we shall call a \emph{finite PDB} here, as a probability
distribution over a finite set of database instances of the same schema. A very
important basic class of finite PDBs is the class of \emph{tuple-independent}
finite PDBs, in which all facts, that is, events of the form \enquote{tuple $t$
appears in relation $R$}, are assumed to be stochastically independent. This
independence assumption implies that the whole probability distribution of the
PDB is fully determined by the marginal distributions of the individual facts
and that the probability of all instances can easily be calculated from these
marginal probabilities. Thus, a tuple-independent PDB can be represented as a
table (resp. as tables) of all possible facts annotated with their respective
marginal probabilities.  According to \cite{Lukasiewicz+2016}, well-known
systems that operate under the assumption of tuple-independence are, among
others, Google's Knowledge Vault \cite{Dong+2014}, the NELL project
\cite{Mitchell+2015} and DeepDive \cite{Niu+2012}.  The focus on
tuple-independent finite PDBs can be further justified by the fact that all
finite PDBs can be represented by first-order (or relational calculus) views
over tuple-independent finite PDBs (see~\cite{Suciu+2011}).\par

Modeling uncertainty by finite PDBs entails an implicit \emph{closed-world
assumption} (CWA) \cite{Reiter1978}:
\begin{itemize}
\item entities not appearing in the finitely many instances with positive 
       probability do not exist and
\item facts not appearing in these instances are strictly impossible, rather
       than just unlikely.
\end{itemize}
In tuple-independent finite PDBs, this means that facts that are not explicitly
listed with a positive probability are impossible. As has already been argued
by Ceylan, Darwiche, and Van den Broeck \cite{Ceylan+2016}, operating under the
CWA can be problematic. For example, consider a database that collects
temperature measurements in the author's offices. Due to unreliable sensors,
these measurements are inherently imprecise and the database may be regarded as
uncertain and modeled as a PDB. Now suppose that the database never records a
temperature between \SI{20.2}{\celsius} and \SI{20.5}{\celsius}. \emph{Is it
reasonable to derive that such a temperature is impossible?} Or suppose that
the data show that the temperature in the first author's office is always at
least \SI{0.1}{\celsius} below the temperature of the second author's office.
\emph{Should we conclude that it is impossible that the temperature in the
first author's office is higher than the temperature in the second author's
office?} Given the uncertainty of the data, we would rather say it is unlikely
(has low probability, where of course the exact probability depends on the
distribution modeling the uncertainty in the data). Moreover, we would expect
that the event \enquote{the temperature in the first author's office is
\SI{0.05}{\celsius} below that in the second author's office} has a higher
probability than the event \enquote{the temperature in the first author's
office is \SI{10}{\celsius} above that in the second author's office}. In a
closed-world model however, both events have the exact same
probability~$0$.\par

Considerations like these led Ceylan et al. \cite{Ceylan+2016} to proposing a
model of open-world probabilistic databases. Their model is tuple-independent,
but instead of probability $0$, facts not appearing in the database are assumed
to have a small positive probability (below some threshold $\lambda$). However,
Ceylan et al. still assume that all entities in the database come from a fixed
finite universe. Hence their model is \enquote{open-world} with respect to
facts, but not with respect to entities or values.\par

In this work, we take one step further and propose a model of \enquote{truly}
open-world probabilistic databases modeled by an infinite supply of entities.
Formally, we define a \emph{probabilistic database (PDB)} to be a probability
space over a sample space consisting of database instances of the same schema
and with entities from the same infinite universe. Note that every instance in
such a PDB is still finite; it is only the probability space and the universe
of potential entities that may be infinite.\par

There are various ways in which such probabilistic databases may arise in
practice: collecting data from unreliable sources, completing incomplete
databases by using statistical or machine learning models, or even having
datasets entirely represented by machine learning models.  This is not very
different from finite PDBs, except that often it is more natural to allow
infinite domains, for example for numerical values or for strings. One may
argue that in practice the domains are always finite (such as 64-bit integers,
64-bit floating point numbers, or strings with fixed maximum length), but
conceptually it is still much more natural to use models with an idealized
infinite domain, as it is common in most other areas of computer science and
numerical mathematics.\par

In this paper, we explore the mathematical foundations of infinite (relational)
probabilistic databases. The general definition of PDBs is given and discussed
in \cref{sec:pdb}. In \cref{sec:ti} we consider tuple-independent infinite PDBs
and show how to construct a countable, tuple-independent PDB from specified
fact probabilities. Unfortunately, the nice result, that every finite PDB can
be represented by a finite tuple-independent PDB does \emph{not} carry over to
infinite PDBs (\cref{pro:FO-def}). In addition to our investigation of
tuple-independence, we provide an extension of their existence results results
to the practically important block-independent-disjoint PDBs
(\cref{thm:bidnecsuff}). Next, in \cref{sec:completions}, we study the
\enquote{open-world} aspect of PDBs. We start from a given discrete PDB and
construct a countable \enquote{completion} that specifies probabilities for
\emph{every} imaginable instance. The key requirement for such a completion to
be reasonable is that the probability measure is faithfully extended: the new
probability measure should coincide with the old one, when conditioned over old
instances. We extend the construction of countable tuple-independent PDBs to a
construction of tuple-independent completions (\cref{thm:indfactcompl}). Albeit
query evaluation is not the focus of this paper, in \cref{sec:query-evaluation}
we hint that it is algorithmically not completely out of reach even in the
infinite setting. Using a na\"{i}ve truncation procedure, we show how to lift
query evaluation for finite PDBs to obtain approximate query answer
probabilities in the case of countably infinite tuple-independent PDBs. Note,
that this is only the very first step towards the algorithmic investigation of
our infinite PDBs.

% ┌────────────┬──────────────────────────────────────────────────────────────┐
% │ Subsection │ Related Work                                                 │
\subsection*{Related Work}
We rely foundationally on the extensive work on finite PDBs (see, for example
\cite{Aggarwal+2009,Suciu+2011,VandenBroeck+2017}). Although some system-oriented
approaches are capable of dealing with continuous PDBs (like MCDB
\cite{Jampani+2008}, PIP \cite{Kennedy+2010}, ORION \cite{Singh+2008} and the
extended Trio system \cite{Agrawal+2009}), these systems typically model
probabilistic databases with an a priori bounded number of facts. In the case
of probabilistic XML \cite{Abiteboul+2009,Kimelfeld+2013} (that is,
probabilistic tree databases) a continuous extension with solid theoretical
foundations has been proposed \cite{Abiteboul+2011}, which however also only
allows a bounded number of facts (resp. \emph{leaf nodes}) in its instances. On
the other hand, a proposed extension of probabilistic XML that allows for
unbounded tree structures does not account for continuous distributions
\cite{Benedikt+2010}.\par

Work on \emph{incomplete databases} \cite{Abiteboul+1995,Green2009,
Imielinski+1984,vanderMeyden1998}, which is also an important source of
motivation for our work, has always naturally assumed potentially infinite
domains, but not treated them probabilistically.\par

In the context of probabilistic databases we want to emphasize again the impact
of the OpenPDB model \cite{Ceylan+2016} to our investigations. More recently,
it has been proposed to extend OpenPDBs using domain knowledge in the form of
ontologies \cite{Borgwardt+2017,Borgwardt+2018}, yielding more intuitive query
results with respect to the open-world assumption in OpenPDBs.\par

In the AI community and in probabilistic programming, open-universe models have
been considered before. Inference in probabilistic models is closely related to
query answering in probabilistic databases \cite{VandenBroeck+2017}. In this
area, some related work has been conducted, although with different backgrounds
and aims. Languages respectively models like BLOG \cite{Milch+2005}, Infinite
Domain Markov Logic \cite{Singla+2007}, probabilistic logic programming
\cite{DeRaedt+2016}, and Probabilistic Programming Datalog \cite{Barany+2017}
are capable of describing infinite probability spaces of structures. It is also
worth mentioning, that weighted first order model counting has been previously
considered in an open-universe setting with given, relation-level probabilities
\cite{Belle2017}.\par

Finally, let us point out a fundamental difference between our notion of
\emph{countable} tuple-independent PDBs and notions of limit probabilities in
asymptotic combinatorics (for example, \cite{Bollobas2001,Spencer2001}). For
example, the classical Erd\H{o}s-R\'enyi model $\scr G(n,p)$ of random graphs
is also what we would call a tuple-independent model: the edges of an
$n$-vertex graph are drawn independently with probability $p$. However, the
sample space is finite, it consists of all $n$-vertex graphs. Then the behavior
of these spaces as $n$ goes to infinity is studied. This means that the
properties of very large graphs dominate the behavior observed here. This
contrasts our model of infinite tuple-independent PDBs, which is dominated by
the behavior of PDBs whose size is close to the expected value (which for
tuple-independent PDBs is always finite). Both views have their merits, but we
believe that for studying probabilistic databases our model is better suited.

% ╒═════════╤═════════════════════════════════════════════════════════════════╕
% │ Section │ Preliminaries                                                   │
% ╘═════════╧═════════════════════════════════════════════════════════════════╛
\section{Preliminaries}
By $\ds N$ we denote the set of positive integers, and by $\RR$ the set of real
numbers. We denote open, closed and half-open intervals of reals by
$(r,s),[r,s], [r,s),(r,s]$. If $M$ is a set, then $2^M$ denotes the power set
of $M$, that is, the set of all subsets of $M$.

% ┌────────────┬──────────────────────────────────────────────────────────────┐
% │ Subsection │ Relational Databases and Logic                               │
\subsection{Relational Databases and Logic}\label{ssec:pre-dbl}
We start out by introducing basic notions of relational databases and logic
(see \cite{Abiteboul+1995}), leading us towards the definition of the standard
model of probabilistic databases of \cite{Suciu+2011} as it will be introduced
in \cref{sec:pdb}.\par

We fix an arbitrary (possibly uncountable) set $\UU$ to be the \emph{universe}
(or \emph{domain}). A \emph{database schema} $\tau = \{R_1,\dots,R_m\}$
consists of relation symbols where each relation symbol $R\in\tau$ has an
associated \emph{arity} $\ar(R)\in\NN$. A \emph{database instance} $D$ of
schema $\tau$ over $\UU$ (for short: \emph{$(\tau,\UU)$-instance}) consists of
\emph{finite} relations $R^D\subseteq\UU^{\ar(R)}$ for all $R\in\tau$. We
denote the set of all $(\tau,\UU)$-instances by $\dbinst$.\par

In terms of logic, a $(\tau,\UU)$-instance is hence a relational structure of
vocabulary $\tau$ with universe $\UU$ in which all relations are finite.\par

It will often be convenient for us (and is quite common in database theory) to
identify database instances as collections of \emph{facts} of the form
$R(a_1,\dots,a_k)$ where $R\in\tau$ is $k$-ary and $(a_1,\dots,a_k)\in \UU^k$.
By $\Facts$ we denote the set of all facts of schema $\tau$ with universe
$\UU$. Then $\dbinst$ is the set of all finite subsets of $\Facts$.  The size
$\lVert D\rVert$ of an instance $D\in\dbinst$ is the number of facts it
contains, that is, $\lVert D\rVert = \sum_{R\in\tau}\lvert R^D\rvert$. The
\emph{active domain} $\adom(D)$ of a $(\tau,\UU)$-instance is the set of all
elements of $\UU$ occurring in the relations of $D$.\par

We use standard first-order logic $\FO$ over our relational vocabulary $\tau$,
which we may expand by constants from $\UU$. By $\FO[\tau,\UU]$ we denote the
set of all first-order formulas of vocabulary $\tau\cup\UU$. Note that in our
notation we do not distinguish between an element $a\in\UU$ and the
corresponding constant and hence between a fact $R(a_1,\ldots,a_k)$ and the
corresponding atomic first-order formula. For an \FO{}-formula
$\phi(x_1,\dots,x_k)\in \FO[\tau,\UU]$ with free variables $x_1,\dots,x_k$ and
an instance $D\in \dbinst$, by $\phi(D)$ we denote the set of all tuples
$(a_1,\dots,a_k)\in\UU^k$ such that $D$ satisfies $\phi(a_1,\dots,a_k)$
(written as $D\models \phi(a_1,\dots,a_k)$).\par

\begin{fact}\label{fact:adom}
Suppose that\/ $\UU$ is infinite. Let $\phi$ be an $\FO$-formula with $k$ free
variables, i.\,e. $\phi(x_1,\dots,x_k)\in\FO[\tau,\UU]$ and let $D\in\dbinst$
such that $\phi(D)$ is finite. Then $\phi(D) \subseteq (\adom(D)\cup
\adom(\phi))^k$, where $\adom(\phi)$ denotes the set of all constants from
$\UU$ occurring in $\phi$.
\end{fact}

A \emph{view} of source schema $\tau$ and target schema $\tau'$ is a mapping
$V\colon\dbinst\to\dbinst[\tau',\UU]$. A ($k$-ary) query is a view $Q$ whose
target schema consists of a single ($k$-ary) relation symbol $R_Q$. Slightly
abusing notation, we usually denote the relation $\smash{R_Q^{Q(D)}}$ of the
image of an instance $D$ under $Q$ by $Q(D)$. For $0$-ary (\emph{Boolean})
queries, we identify the answer $\emptyset$ with \false{} and the answer
$\{()\}$ with \true{}. We defined queries in terms of views, but of course
views can also be regarded as finite sets of queries.\par

A view $V\colon\dbinst\to\dbinst[\tau',\UU]$ is an \emph{$\FO$-view} if for
every $k$-ary relation symbol $R\in\tau'$ there exists a first order formula
$\phi_R(x_1,\dots,x_k)\in\FO[\tau,\UU]$ such that for all $D\in\dbinst$ it
holds that $R^{V(D)} = \phi_R(D)$. 

% ┌────────────┬──────────────────────────────────────────────────────────────┐
% │ Subsection │ Series and Infinite Products                                 │
\subsection{Series and Infinite Products}
In the analysis of independence in infinite probabilistic databases, infinite
products naturally occur. Therefore, we summarize a few important classical
results from the theory of infinite products in the following. For details we
refer the reader to chapter 7 of \cite{Knopp1996}.\par

Let $(x_i)_{i\geq 1}$ be a sequence of real numbers. Consider the series
$\sum_{i\geq 1} x_i$. If the range of the summation is clear, we might simply
write $\sum_i x_i$. The \emph{value} of $\sum_i x_i$ is the limit
$\lim_{n\to\infty}\sum_{i=1}^n x_i$ of its partial sums, given that this limit
exists. $\sum_i x_i$ \emph{converges}, if its value is existent and finite (and
\emph{diverges} otherwise). The series is called \emph{absolutely convergent},
if $\sum_i\lvert x_i\rvert$ converges. Being absolutely convergent is
equivalent to the condition that the value of the series is invariant to
reorderings of its summands.\par

An infinite product $\prod_i x_i$ \emph{converges} if there exists $i_0$ such
that 
\begin{equation}\label{eq:i0}
\lim_{n\to\infty}\sprod_{i=i_0}^n x_i
\end{equation}
exists, is finite and non-zero (and \emph{diverges} otherwise). Note, how this
definition differs from the definition of convergence of a series; see
\cite{Knopp1996} for the technical rationale. The \emph{value} of $\prod_i x_i$
is given by \eqref{eq:i0} for $i_0=0$, if it exists. Note that in particular,
diverging products may have value $0$ (which is the case if all the limits
\eqref{eq:i0} are $0$) but also converging products may have value $0$ (which
happens whenever the product contains a finite number of $0$s and the rest of
the product converges).\par

A necessary condition for infinite products to converge is that its factors
approach $1$. In analogy to series, where the corresponding criterion is that
the summands approach $0$, infinite products are commonly written in the form
$\prod_i (1+a_i)$. An infinite product $\prod_i(1+a_i)$ \emph{converges
absolutely}, if $\prod_i(1+\lvert a_i\rvert)$ converges.\par

\begin{fact}[\cite{Knopp1996}, pp. 229 and 234]\label{fac:knopp}\leavevmode
\begin{enumerate}
\item An infinite product $\prod_i(1+a_i)$ converges (absolutely) if and only
       if\/ $\sum_i a_i$ converges (absolutely).
\item An infinite product $\prod_i(1+a_i)$ converges to the same
       value under arbitrary reorderings of its factors if and only if it is
       absolutely convergent.
\end{enumerate}
\end{fact}

Later on, we use arbitrary countably infinite index sets $I$ in the
consideration of infinite products and series. In that case, we fix an
arbitrary order on $I$ for the summation. Since in all of these cases, the
corresponding series will be absolutely convergent, this won't cause any
problems.\par

We will at some point use the following relationship between infinite products
and series. Its proof can be found in the \hyperref[app:23]{appendix}.

\newcounter{techinr}
\begin{lemma}[\cite{Simmons2013}]\label{lem:tech1}
Let $(a_i)_{i\in I}$ be a countably infinite sequence of real numbers such
that $\sum_i a_i$ is absolutely convergent. Then
\begin{equation}\setcounter{techinr}{\value{equation}}\label{eq:subsetsum}
  \prod_{i\in I}\big(1+a_i\big)
= \sum_{\substack{J\subseteq I\\\text{\normalfont finite}}}
  \prod_{i\in J} a_i
\end{equation}
and both sides of \eqref{eq:subsetsum} are absolutely convergent.
\end{lemma}

% ┌────────────┬──────────────────────────────────────────────────────────────┐
% │ Subsection │ Probability Theory                                           │
\subsection{Probability Theory}\label{ssec:pre-pdb}
We review a few basic definitions from probability theory. Recall that a
\emph{$\sigma$-algebra} over a set $\Omega$ is a set $\fr A\subseteq
\SubsetsOf\Omega$ such that $\Omega\in\fr A$ and $\fr A$ is closed under
complementation and countable unions. A \emph{probability space} is a triple
$\scr S = (\Omega,\fr A,P)$ consisting of

\begin{itemize}
\item a non-empty set $\Omega$ (the \emph{sample space});
\item a $\sigma$-algebra $\fr A$ on $\Omega$  (the \emph{event space});
       and
\item a function $P\colon\fr A\to [0,1]$ (the \emph{probability
       measure}) satisfying
       \begin{enumerate}
       \item $P(\Omega) = 1$ and
       \item for every sequence $A_1,A_2,\ldots$ of 
       mutually disjoint events $A_i\in\fr A$ ($i\geq 1$)
         \[P\vleft(\bigcup_{i\ge 1}A_i\vright) = \sum_{i\geq 1}P(A_i)\text.\]
       \end{enumerate}
       (condition (2) is called \emph{$\sigma$-additivity}).
\end{itemize}

A probability space is called \emph{discrete} or \emph{countable}, if its
sample space $\Omega$ is at most countably infinite, and \emph{uncountable}
otherwise. It is called \emph{finite}, if $\Omega$ is finite. In discrete
probability spaces, $\fr A$ is usually the power set $\SubsetsOf{\Omega}$.
Then, defining $P(\{\omega\})$ for all $\omega \in \Omega$ already completely
determines the whole probability distribution due to the $\sigma$-additivity of
$P$.\par

If the components of a probability space $\scr S$ are anonymous,
$\Pr_{S\sim\scr S}$ is the probability distribution of the random variable
associated with drawing a sample from $\scr S$ (and we may omit the subscript,
if it is clear from the context). We let $\Ex(X)$ denote the expectation of a
random variable (RV) $X$.

\begin{example}\label{exa:ps}
Suppose we take a universe $\UU=\Sigma^*\cup\RR$, where $\Sigma$ is a finite
alphabet, as our sample space. To define a $\sigma$-algebra $\fr A$ on $\UU$,
we let $\fr A_1:=\SubsetsOf{\Sigma^*}$ and let $\fr A_2$ be a standard
$\sigma$-algebra over the reals $\RR$, say, the Borel sets. Then we let
$\mathfrak A$ be the set of all sets $A\subseteq \UU$ such that $A\cap\RR\in\fr
A_2$ (note that we automatically have $A\cap \Sigma^*\in\fr A_1$). To define a
probability distribution $P$, we take a distribution $P_1$ on $\fr A_1$, for
example the distribution defined by
\begin{equation*}
P_1\big(\{w\}\big)\coloneqq\sfrac{6}{\pi^2 (n+1)^2\lvert\Sigma\rvert^n}
\end{equation*}
for all words $w\in\Sigma^*$ of length $\lvert w\rvert =n$ and a distribution
$P_2$ on $\fr A_2$, say, the normal distribution $N(0,1)$ with mean $0$ and
variance $1$, and let $\smash{P(A)\coloneqq\sfrac{1}{2}P_1(A\cap\Sigma^*)+
\sfrac{1}{2}P_2(A\cap\RR)}$ for all $A\in\fr A$.\par

Note that in the definition of $P_1$ we use that $\smash{\sum_{n\ge 1}
\sfrac{1}{n^2}=\sfrac{\pi^2}{6}}$.
\end{example}\par\bigskip

A collection $(A_i)_{i\in I}$ (with arbitrary index set $I$) of events of a
probability space $(\Omega,\fr A,P)$ is called \emph{independent}, if
\[P\vleft(\bigcap_{i\in M} A_i\vright) = \prod_{i\in M}P(A_i) \qquad\text{for
every finite $M\subseteq I$.}\]\par 

If $(A_i)_{i\in I}$ is independent, then so is the sequence 
$(\ol{A_i})_{i\in I}$ of the complements of the $A_i$. If $(A_1,A_2,\dots)$ is
a \emph{countably infinite} sequence of independent events, then
\[P\vleft(\bigcap_{i\geq 1} A_i\vright) = \prod_{i\geq 1} P(A_i)\text.\]\par

We use a variant of an important classical result, known as the \emph{(Second)
Borel-Cantelli Lemma} (see, for example, \cite{Fristedt+1997}).

\begin{lemma}\label{lem:bc}
If $(\Omega,\fr A,P)$ is a probability space and $A_1,A_2,\dots$ a
sequence of pairwise independent events. If $\sum_{i\geq 1}P(A_i)=\infty$,
then \[P\vleft(\bigcap_{i\geq 1}\bigcup_{j\geq i} A_i\vright) = 1\text,\]
that is, the probability that infinitely many events $A_i$ occur is $1$.
\end{lemma}

% ╒═════════╤═════════════════════════════════════════════════════════════════╕
% │ Section │ Probabilistic Databases                                         │
% ╘═════════╧═════════════════════════════════════════════════════════════════╛
\section{Probabilistic Databases}\label{sec:pdb}
In the current literature (e.\,g. \cite{Suciu+2011}), probabilistic relational
databases are defined to be probability spaces whose sample space is a finite
set of database instances over the same schema and the same universe. We extend
this notion in a straightforward way to infinite spaces.\par

Let $\UU$ be some set and $\tau$ be a database schema. \emph{We always assume
that the universe $\UU$ implicitly comes with a $\sigma$-algebra $\fr U$.}
Moreover, we assume that $\{u\}\in\fr U$ for all $u\in\UU$. If $\UU$ is
countable, this implies $\fr U = 2^\UU$. A typical uncountable universe is
$\Sigma^*\cup\mkern 1mu\RR$ for some finite alphabet $\Sigma$; we described a
natural construction of a $\sigma$-algebra for this universe in \cref{exa:ps}.
We lift the $\sigma$-algebra $\fr U$ to a $\sigma$-algebra $\fr F$ on the set
$\Facts$ of all facts by a generic product construction. That is, we let $\fr
F$ be the $\sigma$-algebra generated by all sets of the form
\[\big\{R(u_1,\dots,u_k)\colon u_1\in U_1,\dots, u_k\in U_k\big\}\] for $k$-ary
$R\in\tau$ and $U_1,\dots,U_k\in\fr U$. Note that the assumption $\{u\}\in\UU$
for all $u\in\UU$ implies $\{f\}\in\fr F$ for all $f\in\Facts$ and thus $\fr
F=2^{\Facts}$ if $\UU$ is countable.\par

In the following, we refer to the sets $F\in\fr F$ as \emph{measurable} sets of
facts.

\begin{definition}\label{def:pdb}
A \emph{probabilistic database (PDB)} of schema $\tau$ and universe $\UU$ is a
probability space $\dd = (\Omega,\fr A, P)$ such that $\Omega$ is a set of
$(\tau,\UU)$-instances and for all measurable sets $F \subseteq \Facts$ the
event $\cal E_F \coloneqq \{D\in\Omega\colon F\cap D\neq \emptyset\}$ belongs
to $\fr A$.
\end{definition}

Note, that if the universe $\UU$ is countable, then the containment of events
$\cal E_F$ in $\fr A$ is equivalent to the containment of the events $\cal
E_f\coloneqq \cal E_{\{f\}}$ for every fact $f$.\par

A set of database instances of the same schema over the same universe is often
called an \emph{incomplete database} and its elements are referred to
\emph{possible worlds} \cite[ch. 19]{Abiteboul+1995}. This terminology is also
used in the context of probabilistic databases. However, we prefer to call the
elements of the sample space $\Omega$ of a probabilistic database $\dd$ the
\emph{instances} of $\dd$. One reason is that we may have instances
$D\in\Omega$ with probability $0$ in $\dd$. Calling such instances
\enquote{possible worlds} may be misleading. In fact, if the sample space
$\Omega$ is uncountable, we typically have probability $0$ for every single
database instance.\par

Typically, the $\sigma$-algebra $\fr A$ of a PDB $\dd$ will be constructed by
lifting the $\sigma$-algebra on the facts (denoted by $\fr F$ above) to a
generic \footnote{By \enquote{generic} we mean that $\fr A$ will just be
constructed in a standard way from the $\sigma$-algebra $\fr F$. By providing
such a generic construction, we can avoid worrying too much about the
$\sigma$-algebras when specifying PDBs.} $\sigma$\nobreakdash-algebra $\fr A$
on the set $\dbinst$ of all finite subsets of $\Facts$. In probability theory,
probability spaces on finite or countable subsets of a probability space are
known as \emph{point processes} \cite{Daley+2003}. There are standard,
\enquote{product type} constructions for lifting $\sigma$-algebras from a set
to its finite (or countable) subsets. Yet, some issues are particular to the
database setting and require extra care. For example, we may want all
first-order views to be measurable mappings between the corresponding spaces
(cf. \cref{ssec:querysem}).  However, we are not going to delve into these
issues in this paper and refer to future work for details.

\begin{example}
Incomplete databases are often specified by relations with null values. In our
framework, we can conveniently describe a probability distribution on the
\enquote{completions} of an incomplete database.\par

Suppose that our universe is $\Sigma^*\cup\RR$, where $\Sigma$ is a standard
alphabet like ASCII or UTF-8. Further suppose that we have a schema $\tau$ that
contains a $5$-ary relation symbol $R$ with attributes \textit{FirstName},
\textit{LastName}, \textit{Gender}, \textit{Nationality} and \textit{Height}
(in this order).\par

Assume first that in this relation $R$ we have a single null value $\bot$ in a
tuple $(\tt{Peter}, \allowbreak{}\tt{Lindner}, \allowbreak{}\tt{male},
\allowbreak{}\tt{German}, \allowbreak{}\bot)$.  We may assume that the missing
height is distributed according to a known distribution of heights of German
males, maybe a normal distribution with a mean around 180~(cm). This gives us a
probability distribution on the possible completions of our incomplete database
and hence a probabilistic database. Note that this is an uncountable
probabilistic database with a distribution derived from a normal distribution
on the reals.\par

Now assume that we have a null value in the first component of a tuple, for
example $(\bot, \allowbreak{}\tt{Grohe}, \allowbreak{}\tt{male},
\allowbreak{}\tt{German}, \allowbreak{}\tt{183})$. Again, we may complete it
according to some distribution on $\Sigma^*$. To find this distribution, we may
take a list of German names together with their frequencies. However, there may
be a small probability that the missing name does not occur in the list. We can
model this by giving a small positive probability to all strings not occurring
in the list, decaying with increasing length. Again, this would give us a
probabilistic database, this time a countable one.\par

If we have several null values, we can assume them to be independent and
complete each of them with its own distribution. This independence assumption
can be problematic, especially, if we have two null values in the same tuple.
For example if the above tuples would additionally list the birth year and the
year of graduation, we would want the birth year to refer to an earlier point
in time than the year of graduation. If we do not want to make an independence
assumption, we can directly define the joint distribution on the completions
of all missing values.\par

Note that this example is related to recent work of Libkin~\cite{Libkin2018},
in which probabilistic completions of incomplete databases are studied in terms
of limit probabilities as the size of the universe goes to infinity.
\end{example}

We call a PDB \emph{finite\,/\,\allowbreak{}discrete\,/\,\allowbreak{}%
uncountable}, if its underlying probability space is \emph{finite\,/\,%
\allowbreak{}discrete\,/\,\allowbreak{}uncountable}. Note especially, that
these notions refer to the cardinality of the sample space rather than to the
size of individual instances (which is in our framework always finite).
Sometimes (in particular in \cref{sec:ti}), we will use the term
\enquote{countable} in a looser sense for PDBs that may have an uncountable
universe, but where the probability distribution is completely determined by
the probabilities of countably many facts (and hence a countable
\enquote{sub-PDB}).\par

% ┌────────────┬──────────────────────────────────────────────────────────────┐
% │ Subsection │ Queries and Views                                            │
\subsection{Queries and Views}\label{ssec:querysem}
In this section, we define the semantics of queries and views applied to
probabilistic databases. Let $\dd=(\Omega,\fr A,P)$ be a PDB of schema $\tau$
with universe $\UU$ and let $V$ be a view of source schema $\tau$ and target
schema $\tau'$. For simplicity, let us first assume that $\dd$ is countable.
Then we let $\dd'\coloneqq V(\dd)=(\Omega',\fr A',P')$ be the PDB of schema
$\tau'$ defined via

\begin{equation}\label{eq:1}
P'\big(\{D'\}\big) \coloneqq P\big(V^{-1}(D')\big)
\end{equation}

for every $D'\in\Omega'$ where $\Omega'$ is the image of $V$ on $\Omega$.\par

In the general case, let $\fr A'$ be a $\sigma$-algebra on
$\dbinst[\tau',\UU]$.  Assume that the mapping $V$ is \emph{measurable} with
respect to $\fr A$ and $\fr A'$, that is, $V^{-1}(A')\in\fr A$ for all
$A'\in\fr A'$. Then, $\dd' \coloneqq V(\dd) = (\Omega',\fr A',P')$ is the PDB
of schema $\tau'$ defined by

\begin{equation}\label{eq:2}
P'\big(A'\big)\coloneqq P\big(V^{-1}(A')\big)
\end{equation}

for every $A'\in\fr A'$. Since we are mostly interested in countable PDBs in
this paper, we do not want to delve into a discussion of the measurability
condition.\par

The semantics of views defined in \eqref{eq:1} and \eqref{eq:2} yields a
semantics of queries on probabilistic databases as a special case. However, for
queries $Q$ one is often interested in the marginal probabilities of individual
tuples in the query answer,
\begin{equation*}
\Pr_{D\sim\dd}\big(\vec a\in Q(D)\big)\text.
\end{equation*}
Usually, this marginal probability is only of interest in countable PDBs.

% ┌────────────┬──────────────────────────────────────────────────────────────┐
% │ Subsection │ Size Distribution                                            │
\subsection{Size Distribution}\label{sec:size}
Let $\dd$ be a probabilistic database of schema $\tau$ with universe $\UU$.
Let $S_\dd$ be the random variable that associates with each instance
$D\in\dbinst$ its size $\lVert D\rVert$, that is, the number of facts that $D$
contains. Observe that if $\dd$ is countable then the expected size of an
instance of $\dd$ is \footnote{$\smash{S_\dd}$ is the sum of the indicator RV
associated with the events $\smash{\cal E_f}$. The expectation of
$0$-$1$-valued RV is equal to their probability to take the value $1$. Finally
note, that linearity of expectation holds for countably infinite sums of
$[0,\infty)$-valued RV \cite[p.~50]{Fristedt+1997}.}%

\begin{equation}\label{eq:3}
  \Ex\big(S_\dd\big)
= \smashoperator{\sum_{f\in F[\tau,\UU]}}\;
  {\textstyle\Pr_{D\sim\scr D}}\big(D\in\cal E_f\big).
\end{equation}

For uncountable PDBs, the sum in \eqref{eq:3} is replaced by an integral. It is
easy to construct examples of (countable) PDBs where $\Ex(S_\dd) = \infty$.\par

\begin{example}\label{ex:inf-exp}
Let $\tau=\{R\}$ with a unary relation symbol $R$ and $\UU=\NN$. For every
$n\geq 1$, let $\smash{p_n\coloneqq \sfrac{6}{\pi^2n^2}}$ (so $\sum_n p_n = 1$)
and let $D_n$ be a $(\tau,\UU)$-instance with
$R^{D_n}\coloneqq\{1,\dots,2^n\}$.  Then $\lVert D_n\rVert = 2^n$. Define $\dd$
by letting $\Pr_{D\sim\dd}(\{D_n\}) \coloneqq p_n$ and $\Pr_{D\sim\dd}(\{D\}) =
0$ for all $D\in\dbinst- \{D_n\colon n\in\NN\}$.\par

Then $\Ex(S_\dd) = \sum_n p_n\lVert D_n\rVert = \sum_n
\sfrac{6\cdot 2^n}{\pi^2n^2} = \infty$.
\end{example}\par

Probabilistic databases with infinite expected instance size may not be the
most relevant in practice. We will see later that tuple-independent PDBs always
have a finite expected size. While the expected instance size of a PDB can
be infinite, the probability that it is large goes to zero:

\begin{equation}\label{eq:4}
\lim_{n\to\infty}\Pr(S_{\scr D}\geq n) = 0\text.
\end{equation}

To see this, just consider the decreasing sequence of events
\begin{equation*}
A_n\coloneqq\big\{D\colon\lVert D\rVert\geq n\big\}
\end{equation*}

and let $A\coloneqq\bigcap_n A_n$. Then $A=\emptyset$, because a PDB only
contains finite instances. Thus $\lim_{n\to\infty}\Pr(A_n)=\Pr(A)=0$.

A consequence of this observation is the following useful proposition.

\begin{proposition}\label{pro:size}
Let $\dd$ be a (possibly uncountable) PDB. Then the set $F_\omega$ of all facts
$f$ with probability $p_f\coloneqq\Pr_{D\sim\dd}(D\in\cal E_f)>0$ is countable.
\end{proposition}

\begin{proof}
For every $k\in\NN$ we let $F_k$ be the set of all facts $f$ with $p_f> 1/k$.
Then $F_\omega=\bigcup_k F_k$. We claim that for all $k$ the set $F_k$ is
finite; this will imply that $F$ is countable.\par

To prove the claim, let $k\in\NN$. Suppose for contradiction that $F_k$ is
infinite and let $f_1,f_2,\dots\in F_k$. By \eqref{eq:4}, there is an $n$ such
that $\Pr(S_\dd> n)< \smash{(2k)^{-1}}$. Choose such an $n$. For every
$i\in\NN$ let $X_i$ be the indicator random variable of the event
\enquote{$S_\dd(D)\leq n$ and $f_i\in D$} and let $Y_i\coloneqq\sum_{1\leq
j\leq i} X_j$. Then

\begin{align*}
\Pr\big(X_i=1\big) &{}= 1-\Pr\big(S_\dd(D)> n\cup f_i\notin D\big)\\
                   &{}> 1-\sfrac{1}{2k}-\big(1-\sfrac{1}{k}\big)
                      = \sfrac{1}{2k}\text,
\end{align*}

and thus $\Ex(Y_i)> i\cdot\smash{(2k)^{-1}}$. In particular, $\Ex(Y_{2kn}) >
n$, which implies that $\Pr(Y_{2kn}> n) > 0$. However, every instance $D$ with
positive $Y_{2kn}(D)$ satisfies $S_\dd(D)\leq n$ and therefore $\lVert
D\rVert\leq n$, but contains (since $Y_{2kn}>n$) at least $n$ of the facts
$f_1,f_2,\dots,f_{2kn}$. This is a contradiction.
\end{proof}

% ╒═════════╤═════════════════════════════════════════════════════════════════╕
% │ Section │ Tuple-Independence in the Infinite                              │
% ╘═════════╧═════════════════════════════════════════════════════════════════╛
\section{Tuple-Independence in the Infinite}\label{sec:ti}
With the above framework in mind, we turn our attention to an infinite
extension of the idea of tuple-independence. This is motivated by the major
importance of tuple-independence in the traditional finite setting. As we will
see, the notions will be more involved and more fundamental questions have to
be addressed. For the following discussion let $\dd = (\Omega, \fr A, P)$ be a
probabilistic database and let $\fr F$ be a suitable $\sigma$-algebra on the
set of all facts, whose elements we call \emph{measurable} sets of facts.
Recall from \cref{sec:pdb} that $\cal E_f$ denotes the event
\enquote{the fact $f$ occurs in a randomly drawn instance}. Finite
probabilistic databases are referred to as \enquote{tuple-independent} if all
these events are independent. In consideration of \emph{infinite} probabilistic
databases, we want to broaden that notion, using the events $\cal
E_F=\bigcup_{f\in F}\cal E_f$. Recall that the definition of PDBs requires
that $\cal E_F\in\mathfrak A$ for all measurable sets $F$ of facts.

\begin{definition}
$\dd$ is called \emph{tuple-independent (\ti)} if for all collections $\cal F$
of pairwise disjoint measurable sets of facts the events $\cal E_F$ are
independent, that is, if
\begin{equation}\label{eq:generalti}
  P\vleft(\bigcap_{F\in\cal F'}\cal E_F\vright)
= \prod_{F\in\cal F'} P(\cal E_F)
\end{equation}
for all finite $\cal F'\subseteq \cal F$.
\end{definition}

Observe that this definition matches the definitions from the literature when
applied to a countable setting.

\begin{lemma}\label{lem:oldtinotion}
A \emph{countable} PDB $(\Omega,\fr A,P)$ is tuple-independent if and only if
all events $\cal E_f$ are independent.
\end{lemma}

\begin{proof}
Let $\cal F$ be a collection of disjoint measurable fact sets. It has to be
shown that the events $(\cal E_F)_{F\in\cal F}$ are independent (i. e.,
\eqref{eq:generalti} holds). To see this, we show the independence of
$(\ol{\cal E_F})_{F\in\cal F'}$ where $\ol{\cal E_F}=\Omega-\cal E_f$. Using
the independence of the events $\cal E_f$ (respectively $\ol{\cal E_f}$) and
the fact that $\ol{\cal E_F} = \bigcap_{f\in F} \ol{\cal E_f}$, we have

\begin{equation*}
     P\vleft(\smashoperator[r]{\bigcap_{F\in \cal F'}} \ol{\cal E_F}\vright)
= P\vleft(\bigcap_{F\in \cal F'}\bigcap_{f\in F} \ol{\cal E_f}\vright)
= \prod_{F\in\cal F'}\smash{\underbrace{\prod_{f\in F} P(\ol{\cal E_f})}_
  {\mathclap{=P(\bigcap_{f\in F}\ol{\cal E_f})}}}
= \smashoperator[r]{\prod_{F\in\cal F'}} P(\ol{\cal E_F})\text.
\end{equation*}

\end{proof}

Nevertheless, the definition can in this form be applied to uncountable PDBs,
although raising multiple issues that keep this extension from being
straightforward. In particular, the above lemma does not carry over to a
general uncountable setting as the events $\cal E_F$ are not necessarily
expressible in terms of the events $\cal E_f$ anymore using only countable
union and complementation.

% ┌────────────┬──────────────────────────────────────────────────────────────┐
% │ Subsection │ Construction                                                 │
\subsection{Construction}\label{ssec:construction}
Tuple-independence is a convenient setting for finite PDBs as it suffices to
specify probabilities for all possible facts to obtain a tuple-independent PDB.
In this subsection, we investigate whether the same approach works for infinite
tuple-independent PDBs. From that investigation, we will obtain a sufficient
criterion for the existence of \emph{countable} tuple-independent PDBs. We
revisit the uncountable setting towards the end of the subsection.\par

Let us consider a schema $\tau$ and a universe $\UU$. We let $\Omega \coloneqq
\dbinst$, and we let $\mathfrak A\subseteq 2^{\Omega}$ be an arbitrary
$\sigma$-algebra that contains all events $\cal E_F$ for measurable
$F\subseteq\Facts$.  In fact, for the construction here we can simply let $\fr
A\coloneqq 2^{\Omega}$.  Moreover, we assume that we are given a family
$(p_f)_{f\in\Facts}$ of numbers $p_f\in[0,1]$. The question we ask is: can we
construct a tuple-independent PDB $\dd = (\Omega,\fr A,P)$ such that $P(\cal
E_f) = p_f$ for all $f$?  We will see, that the question can be positively
answered whenever every countable sum of numbers $p_f$ is finite, that is,

\begin{equation}\label{eq:7}
  \sum_{f\in F} p_f<\infty\qquad\text{for every countable }F\subseteq\Facts
  \text.
\end{equation}

In the following, we say that $\sum_{f} p_f$ is \emph{convergent} if
\eqref{eq:7} is satisfied. This is justified by the following argument showing
that if \eqref{eq:7} holds then the set of all facts $f$ with $p_f>0$ is
countable. Thus in this case, up to 0-values, the sum $\sum_fp_f$ is countable.
Indeed, if \eqref{eq:7} holds then for every $k\in\NN$ the set $F_k$ of all $f$
such that $p_f> 1/k$ is finite. Thus the set $F_\omega=\bigcup_kF_k$ of all $f$
such that $p_f> 0$ is countable.  A consequence of this observation is that the
convergence assumption \eqref{eq:7} implies that the resulting PDB will be
countable.\par

In the following, we assume that $\sum_{f} p_f$ is convergent and let
$F_\omega$ be the (countable) set of all $f\in\Facts$ with $p_f>0$. We define a
probability measure $P$ on $(\Omega,\fr A)$ as follows. For $D\in\dbinst$, we
let

\begin{equation*}
P(\{D\}) = \sprod_{f\in D} p_f\sprod_{f\in
    F_\omega\setminus D} \big(1-p_f\big).
\end{equation*}

It follows from \cref{fac:knopp} that this product is well-defined, because
the sum $\sum_{f\in F_\omega} p_f$ is convergent an hence the product
$\prod_{f\in F_\omega} (1-p_f)$ is convergent as well.\par

Note that there are only countably many $D\in\dbinst$ such that $P(\{D\})>0$,
because $P(\{D\})>0$ implies that $D\subseteq F_\omega$, and the countable set
$F_\omega$ has only countably many finite subsets.  Let $\bf D_\omega$ be the
set of all finite subsets of $F_\omega$.  We complete the definition of the
probability measure $P$ by letting $P(A)=\sum_{D\in A\cap\bf D_\omega}P(\{D\})$
for all $A\in\fr A$.\par

The following lemma (which is a variation of a statement proven in
\cite{Renyi2007} by Rényi) ensures that $P$ is a probability measure:

\begin{lemma}[cf. \cite{Renyi2007}, p. 167 et seq.]\label{lem:tech2}
	$P(\Omega) = \sum_{D\in\bf D_\omega} P(\{D\}) = 1$.
\end{lemma}

\begin{proof}
Denote $F_\omega\setminus D$ by $\smash{\ol D}$ for any instance $D$. By
reordering and using \cref{lem:tech1}, we obtain:

\begin{align}
\sum_{D\in\bf D_\omega} P(\{D\})
&{}= \sum_{D\in\bf D_\omega}
                       \sprod_{f\in D} p_f
                       \sprod_{f\in\ol D} (1-p_f)
                                   \notag\\
&{}= \sum_{D\in\bf D_\omega} \sprod_{f\in D} p_f
                       \sum_{\substack{D'\subseteq \ol D\\\text{finite}}}
                       \sprod_{f\in D'} (-p_f)
                       \notag\\
&{}= \sum_{D\in\bf D_\omega} \sprod_{f\in D} p_f
                       \sum_{\substack{D'\in\bf D_\omega\\D'\supseteq D}}
                       \sprod_{f\in D'\setminus D} (-p_f)
                       \notag\\
&{}= \sum_{D\in\bf D_\omega} \sum_{\substack{D'\in\bf D_\omega\\D'\supseteq D}}
                       \sprod_{f\in D} p_f
                       \sprod_{f\in D'\setminus D} (-p_f)
                       \notag\\
&{}= \sum_{D'\in\bf D_\omega} \sum_{D\subseteq D'}
                        \sprod_{f\in D} p_f
                        \sprod_{f\in D'\setminus D} (-p_f)
                        \label{eq:6}\text.
\end{align}

Within \eqref{eq:6}, the summand for $D' = \emptyset$ collapses to a single
empty product and hence equals $1$. If on the other hand $D'\neq\emptyset$, fix
some $f_0\in D'$. It is easy to see that the subsums containing the factor
$p_{f_0}$ and those containing $-p_{f_0}$ exactly cancel each other out. Thus
$P(\Omega) = 1$.
\end{proof}

The previous lemma means that we have indeed constructed a PDB. It remains to
show that $\dd=(\Omega,\fr A,P)$ is tuple-independent and has the right
marginal probabilities for the events $\cal E_f$.

\begin{lemma}\label{lem:constristi}
$\dd$ is \ti{}, and $P(\cal E_f)=p_f$ for all facts $f\in\Facts$.
\end{lemma}

\begin{proof}
By an argument similar to the proof of \cref{lem:oldtinotion}, it suffices
to check the independence of the events $\cal E_f$ for facts $f\in
F_\omega$.\par

Let $F\subseteq F_\omega$ be finite. We prove that $P\big(\bigcap_{f\in F}\cal
E_F\big) = \prod_{f\in F} p_f$. Note that this implies both $P(\cal E_f)=p_f$
\emph{and} the independence of the events $\cal E_f$ for all $f$. Let
$\Omega_F$ denote the set of instances $D\in\bf D_\omega$ with $F\subseteq D$.
We have

\begin{align*}
     P\vleft(\bigcap_{f\in F} \cal E_f\vright)
&{}= \smashoperator{\sum_{D\in\Omega_F}}\; P\big(\{D\}\big)
     \notag\\
&{}= \smashoperator[l]{\sum_{D\in\Omega_F}}
     \sprod_{f\in D} p_f
     \sprod_{f\in F_\omega\setminus D} \big(1-p_f\big)
     \notag\\
&{}= \prod_{f\in F} p_f
     \vleft(
            \sum_{D\in\Omega_F}
            \sprod_{f\in D\setminus F} p_f
            \sprod_{f\in F_\omega\setminus D} \big(1-p_f\big)
     \vright)\text.
\end{align*}

We conclude the proof by showing that the parenthesized term in the last row
equals 1. Note that its products range exactly over all facts in
$F_\omega\setminus F$. Recall that $\ol{\cal E_F}$ is the set of instances that
are disjoint from $F$.

\begin{align*}
   & \smashoperator[l]{\sum_{D\in\Omega_F}}
     \sprod_{f\in D\setminus F} p_f
     \sprod_{f \in F_\omega\setminus D} \big(1-p_f\big)
     \notag\\
={}& \smashoperator[l]{\sum_{D' \in \ol{\cal E_F}}}
     \sprod_{f\in D'} p_f
     \sprod_{\substack{f\in F_\omega\\f\notin F\cup D'}} \big(1-p_f\big)
     \smash{\overbrace{\vleft(
                       \sum_{F'\subseteq F}
                       \sprod_{f\in F'} p_f
                       \sprod_{f\in F\setminus F'} \big(1-p_f\big)
     \vright)}^{=1}}
     \notag\\
={}& \smashoperator[l]{\sum_{D'\in \ol{\cal E_F}}}
     \sum_{F'\subseteq F}
     \sprod_{f\in F'\cup D'} p_f
     \sprod_{\substack{f\in F_\omega\\f\notin F'\cup D'}} \big(1-p_f\big)
     \notag\\
={}& \smashoperator[l]{\sum_{D\in\Omega}}
     \sprod_{f\in D} p_f
     \sprod_{f\in F_\omega\setminus D} \big(1-p_f\big)
     = P(\Omega) = 1\text.\qedhere
\end{align*}

\end{proof}

The above construction, starting from a convergent series of fact probabilities
thus yields a tuple-independent PDB that realizes these given probabilities.
Note also that the given sequence of fact probabilities already determines the
whole probability space. We summarize the result in the following proposition.

\begin{proposition}\label{pro:construction}
Given a family $(p_f)_{f\in\Facts}$ of real numbers $p_f\in[0,1]$ such that
$\sum_f p_f$ is convergent, we can construct a tuple-independent PDB with 
$P(\cal E_f) = p_f$ for all $f\in\Facts$. 
\end{proposition}

Finally, let us briefly discuss the difficulties of obtaining a similar result
for ``truly'' uncountable PDBs, although we refer to future work for a thorough
investigation of that problem. In the countable case discussed above, we
expressed the probabilities of all events $\cal E_F$ using the probabilities
$\cal E_f$ alone. In general, we cannot do this, because all $\cal E_f$ may
have probability $0$. This raises the question even which probabilities should
actually be specified beforehand for constructing the PDB.  We leave it as an
open problem whether uncountable tuple-independent PDBs exist that do not
collapse to discrete probability spaces.

% ┌────────────┬──────────────────────────────────────────────────────────────┐
% │ Subsection │ A Necessary Existence Criterion                              │
\subsection{A Necessary Existence Criterion}\label{ssec:nec}
In the previous subsection, we have seen that the convergence of $\sum_f p_f$
is a sufficient criterion for given fact probabilities to fulfill in order to
ensure the existence of a tuple-independent PDB that is compatible with these
probabilities.  Now, we will prove that this condition is also necessary,
i.\,e. that there is no tuple-independent PDB realizing a divergent series of
fact probabilities.  This result is not limited to the countable case:

\begin{lemma}\label{lem:necessity}
Let $\dd=(\Omega,\fr A,P)$ be a tuple-independent PDB. Then 

\begin{equation*}
\sum_{F\in\cal F}P(\cal E_F) < \infty
\end{equation*}

for all countably infinite collections $\cal{F}\!$ of pairwise disjoint, 
measurable sets of facts.
\end{lemma}

\begin{proof}
Since $\dd$ is tuple-independent, the events $\cal E_F$ are independent.
Suppose $\cal F = \{F_1,F_2,\dots\}$ and for $F=F_i$, let $\cal E_i\coloneqq
\cal E_F$.  Then
\begin{equation*}
 \cal E \coloneqq \limsup_{i\to\infty}\mkern4mu\cal E_i
 = \bigcap_{i\geq 1}\bigcup_{j\geq i}\mkern4mu\cal E_j
\end{equation*}
is the set of instances having a nonempty intersection with infinitely many
sets $F\in\cal F$. Since the sets from $\cal F$ are disjoint and each database
of $\Omega$ has only finitely many facts, $\cal E = \emptyset$ and $P(\cal E) =
0$. By \cref{lem:bc} (the Borel-Cantelli Lemma), this means
$\sum_{F\in\cal F} P(\cal E_F)$ converges.
\end{proof}

Note that in particular, if $\dd$ is a countable tuple-independent PDB (over
database schema  $\tau$ and universe $\UU$), then we have $\sum_{f\in\Facts}
P(\cal E_f) < \infty$. As this sum is exactly the definition of the expected
instance size \eqref{eq:3}, we immediately obtain the following.

\begin{corollary}\label{cor:finitesize}
If\/ $\dd$ is a countable tuple-independent PDB, then its expected instance
size is finite.
\end{corollary}

Finally, we can combine \cref{lem:necessity} and \cref{pro:construction} into
the following characterization of countable tuple-independent PDBs.

\begin{theorem}\label{thm:tinecsuff}
Let $(p_f)_{f\in \Facts}$ with $p_f\in[0,1]$.  There exists a tuple-independent
PDB with fact probabilities $\Pr(\cal E_f)=p_f$ for all $f\in\Facts$  if and
only if $\sum_fp_f$ is convergent.
\end{theorem}

% ┌────────────┬──────────────────────────────────────────────────────────────┐
% │ Subsection │ Definability in Tuple-Independent Probabilistic Databases    │
\subsection{Definability in Tuple-Independent Probabilistic Databases}
\label{sec:expr}%
The viability of \ti{} PDBs in the finite is often justified by the well-known 
result that tuple-independent PDBs are sufficient to describe arbitrary finite 
PDBs by the means of \enquote{\FO-views}.\par

We call a PDB $\dd$ \emph{\FO-definable} over a PDB $\cc$ if there is an
$\FO$-view $V$ such that $\dd = V(\cc)$ (see \cref{ssec:querysem}).
Although not every finite PDB is itself tuple-independent, every finite PDB is
$\FO$-definable over a tuple-independent PDB \cite{Suciu+2011}. Unfortunately,
this result does not extend to infinite PDBs.

\begin{proposition}\label{pro:FO-def}
There is a countably infinite PDB \/$\dd$ that is not \FO-definable over any
tuple-independent PDB.
\end{proposition}

\begin{proof}
Let $\UU\coloneqq\NN$ and $\tau' \coloneqq \{R\}$ for some unary relation
symbol $R$. Let $\dd$ be the database of schema $\tau'$ over $\UU$ defined in
\cref{ex:inf-exp}. Then $\Ex(S_{\dd}) = \infty$, where $S_{\dd}$ was the
random variable associating instances with their size.\par

Suppose for contradiction that $\dd = V(\cc)$ for some \ti{} PDB $\cc$ of some
schema $\tau$. For every $f\in\Facts$ let $p_f \coloneqq \Pr_{C\sim\cc}
(C\in\cal E_f)$. Then, by \cref{cor:finitesize}, $\Ex(S_\cc) = \sum_f
p_f < \infty$. Let $X_\cc$ be the random variable that maps $C\sim\cc$ to
$\lvert\adom(C)\rvert$. Let $k$ be the maximum arity of a relation in $\tau$
and note that for every $(\tau,\UU)$-instance $\lvert\adom(C)\rvert \leq
k\lVert C\rVert$. That is, $X_\cc \leq kS_\cc$.\par

Since $\tau'$ consists of a single unary relation symbol, the view $V$ consists
of a single formula $\phi(x)\in\FO[\tau,\UU]$. Let $c$ be the number of
constants from $\UU$ appearing in $\phi$. By \cref{fact:adom}, for every
$(\tau,\UU)$-instance we have $\lVert V(C)\rVert = \lvert\phi(C)\rvert\leq
\lvert\adom(C)\rvert + c$. But this implies $S_\dd\leq kS_\cc + c$ and
therefore $\Ex(S_\dd) \leq k\Ex(S_\cc) + c < \infty$, a contradiction.
\end{proof}

\begin{remark}
The PDB $\dd$ that we used in the proof of \cref{pro:FO-def} has the property
that the expected instance size is infinite. However, it is not hard to
construct an analogous counterexample with finite expected input size: we
simply construct a PDB $\dd$ where $\Ex(S_\dd)<\infty$ but $\Ex(S_\dd^2) =
\infty$. Instead of the second moment, we can use the $k$th moment for any
$k$.\par

We do not know an example of a PDB $\dd$ with $\Ex(S_\dd^k)<\infty$ such that
$\dd$ is not $\FO$-definable over a tuple-independent PDB. We conjecture,
though, that such an example exists.
\end{remark}

% ┌────────────┬──────────────────────────────────────────────────────────────┐
% │ Subsection │ A Word on Block-Independent Disjoint Probabilistic Databases │
\subsection{A Word on Block-Independent-Disjoint Probabilistic Databases}
\label{sec:bid}%
After studying tuple-independence, we want to turn our attention to a
practically relevant generalization of tuple-independence: the notion of
\emph{block-independent-disjoint (\bid)} PDBs \cite{Suciu+2011}. As
\cite{Lukasiewicz+2016} notes for example, the systems Trio
\cite{Agrawal+2006}, MayBMS \cite{Huang+2009} and MystiQ \cite{Boulos+2005}
realize (finite) PDBs of this category. In such PDBs, the set of all facts is
partitioned into \emph{blocks} of facts with two central properties: first of
all, facts within the same blocks form mutually exclusive events and; second of
all, facts across different blocks are independent. Obviously, the traditional
notion of tuple-independence is the special case of \bid{} PDBs with singleton
blocks.\par 

The usual application of \bid{} PDBs is to incorporate key constraints in PDBs.
Here, we want to provide a more general definition that extends to infinite
settings. In the following, let $\UU$ be some universe and $\tau$ be a database
schema. As usual, we assume that we have suitable $\sigma$ algebra on $\Facts$
and speak of measurable sets of facts.

\begin{definition}\label{def:bid}
Let $\cal B$ be a partition of $\Facts$ into measurable sets. A PDB $\dd =
(\Omega, \fr A, P)$ is \emph{block-independent-disjoint (\bid) with respect to
$\cal B$} or \emph{with blocks $\cal B$}, if

\begin{enumerate}
\item for all $B\in\cal B$ and all disjoint, measurable $B_1,B_2\subseteq B$:
       \[P\big(\cal E_{B_1}\cap \cal E_{B_2}\big) = 0\text,\]
       \label{itm:bidexcl}
\item and for all mutually distinct $B_1,\dots, B_k\in \cal B$ ($k\in\NN$) and 
	all measurable $B_i' \subseteq B_i$ ($1\leq i \leq k$), it holds that
       \[P\vleft(\smashoperator[r]{\bigcap_{1\leq i\leq k}} \cal E_{B_i'}
	       \vright)
       = \smashoperator{\prod_{1\leq i\leq k}} P\big(\cal E_{B_i'}\big)\text.\]
       \label{itm:bidindependence}
\end{enumerate}

A PDB $\dd$ is \emph{block-independent-disjoint (\bid)}, if there exists a
suitable partition $\cal B$ such that $\dd$ is \bid{} with respect to $\cal B$.
\end{definition}

We want to discuss whether and if so, how the previous results generalize from
\ti{} to \bid{} PDBs.\par

First, we note that an analogue of \cref{lem:oldtinotion} holds, which
means that our notion of \bid{} PDBs in a countable setting matches the
traditional definition that only mentions facts:

\begin{lemma}\label{lem:oldbidnotion}
For countable PDBs with blocks $\cal B$, satisfying condition
\eqref{itm:bidexcl} from \cref{def:bid}, condition \eqref{itm:bidindependence}
is equivalent to

\begin{enumerate}
\item[(2')] The sequences $(\cal E_f)_{f\in F}$ are independent for every 
	collection $F$ of facts such that $F$ contains at most one fact from
	each block.
\end{enumerate}

\end{lemma}

The easy proof can be found in the \hyperref[app:412]{appendix}.
\par

Next, we can construct \emph{countable} \bid{} PDBs similarly to the
tuple-independent case.

\begin{proposition}\label{pro:bidconstruction}
Let $\cal B$ be a partition of $\Facts$ into blocks and for every block $B$ let
$\smash{(p_f^B)_{f\in B}}$ such that $\smash{p_f^B\in[0,1]}$ and
$\smash{\sum_{f\in B} p_f^B\leq 1}$. Then we can construct a PDB $\dd =
(\Omega, \fr A,P)$ that is \bid{} with respect to $\cal B$, realizing the given
fact probabilities (i.\,e., $\smash{P(\cal E_f) = p_f^{B(f)}}$, where $B(f)$ is
the block containing $f$) whenever 

\begin{equation}\label{eq:5}
  \sum_{f\in F} p_f^{B(f)} < \infty
\end{equation}

for all countable $F\subseteq\Facts$.
\end{proposition}

\begin{proofsketch}
Let $\Omega\coloneqq\dbinst$ and $\fr A\coloneqq 2^{\Omega}$. If \eqref{eq:5}
holds for all countable $F\subseteq\Facts$, we say $\smash{\sum_{B\in\cal
B}\sum_{f\in B} p_f^{B}}$ converges. This notation is justified like in the
case of tuple-independence.  Similarly to before, it entails that the set
$F_\omega$ of facts with $\smash{p_f^{B(f)}>0}$ is countable. We may thus
suppose that $\cal B$ consists of countably many countable blocks $\cal
B_\omega$ exactly covering the facts $F_\omega$ and that all the remaining
facts are gathered in a single dummy block (this can be found in the proof of 
\cref{lem:oldbidnotion} in the \hyperref[app:412]{appendix}).

\emph{Good} instances contain at most one fact from each $B\in\cal B$ and
\emph{bad} instances violate this condition. We set $\smash{p_\bot^B\coloneqq
1-\sum_{f\in F_\omega\cap B} p_f^B}$ ($p_\bot^B = 1$ for $B$ being the dummy
block) and for every block $B$ and good $D\in\dbinst$, define

\begin{equation*}
\beta(B,D) \coloneqq \begin{dcases}
                       f & \text{if }D\cap B = \{f\}\text,\\
                       \bot & \text{if }D\cap B = \emptyset\text.
                       \end{dcases}
\end{equation*}

We set

\begin{equation*}
P\big(\{D\}\big) \coloneqq
\begin{dcases}
\sprod_{B\in\cal B} p_{\beta(B,D)}^B & \text{if }D\text{ is good,}\\
0                                    & \text{if }D\text{ is bad;}
\end{dcases}
\end{equation*}

and for all other $A\in\fr A$, $P(A) = \sum_{D\in A\cap\bf D_\omega} P(\{D\})$
where $\bf D_\omega$ is the set of finite subsets of $F_\omega$. Analogously to
\cref{ssec:construction}, the required convergence property \eqref{eq:5}
ensures, that this yields indeed a probability measure. Generalizing the proof
of \cref{lem:constristi}, one can show that $\dd$ is indeed a \bid{} PDB.
These two claims are demonstrated in detail in the 
\hyperref[app:413]{appendix}.
\end{proofsketch}

Finally, the necessary condition from \cref{ssec:nec} easily translates
to \bid{} PDBs:

\begin{lemma}\label{lem:bidnec}
Let $\dd = (\Omega,\fr A,P)$ be a \bid{} PDB with blocks $\cal B$. Then, for
every countable collection $(B_i)_{i\geq 1}$ of $\cal B$-blocks and all
measurable subsets $B_i' \subseteq B_i$ (\/$i\geq 1$) it holds that
$\sum_{i\geq 1}P(\cal E_{B_i'}) < \infty$. 
\end{lemma}

\begin{proof}
	This is proven exactly like in the proof of \cref{lem:necessity}
	with $B_i'$ in the role of $F_i$.
\end{proof}

\cref{pro:bidconstruction} and \cref{lem:bidnec} can be combined, yielding the
following characterization of existence for countable \bid{} PDBs:

\begin{theorem}\label{thm:bidnecsuff}
Let $\cal B$ be a partition of facts and for every $B\in\cal B$ let
$\smash{(p_f^B)_{f\in B}}$ be a sequence with $\smash{p_f^B\in[0,1]}$ such that
$\smash{\sum_{f\in B}p_f^B\leq 1}$. There exists a block-independent-disjoint
PDB with fact probabilities $\smash{(p_f^B)_{f\in B,B\in\cal B}}$ if and only
if $\smash{\sum_{B\in\cal B}\sum_{f\in B}p_f^B}$ converges.
\end{theorem}

% ╒═════════╤═════════════════════════════════════════════════════════════════╕
% │ Section │ Completions of Probabilistic Databases                          │
% ╘═════════╧═════════════════════════════════════════════════════════════════╛
\section{Completions of Probabilistic Databases}\label{sec:completions}
Now that we have established a model of infinite independence assumptions, we
want to revisit the open-world assumption in probabilistic databases. We want
to use our construction of countable tuple-independent PDBs to deploy a
\enquote{completed} version of a given PDB. Let $\dd$ be a PDB with sample
space $\Omega\subsetneq\dbinst$ where $\tau$ is a database schema and $\UU$ a
universe. Typically (but not necessarily), we think of $\Omega$ being finite.
Our construction shall expand the sample space $\Omega$ to all of $\dbinst$. In
order to obtain results that are consistent with the original data from $\dd$,
this expansion should preserve the basic structure of the probability space
$\dd$, that is, its internal correlations and the proportions of already known
fact probabilities.\par

\begin{definition}\label{def:completion}
Let $\dd=(\Omega,\fr A,P)$ be a PDB with $\Omega\subsetneq\dbinst$. A
\emph{completion} of $\dd$ is a PDB $\dd'=(\Omega',\fr A',P')$ with
$\Omega'=\dbinst$ and $\fr A'\supseteq \fr A$ such that $P'(\Omega) > 0$ and
for all $A\in\fr A$, the following \emph{completion condition} holds:

\begin{gather}\label{eq:cc}\tag{CC}
P'\big(A\bigunder\Omega\big) = P(A)\text.
\end{gather}

When considering a completion $\dd'$ of $\dd$, we refer to $\dd$ (and its
components) as \emph{original}.
\end{definition}

\begin{remark}
Applying the closed-world-assumption to a PDB corresponds to considering the
completion that sets all probabilities of new instances to $0$.
\end{remark}

\begin{remark}
Although we use similar notions, the completions of \cref{def:completion} are
not directly related to the concept of completion of measure spaces in measure
theory.
\end{remark}

% ┌────────────┬──────────────────────────────────────────────────────────────┐
% │ Subsection │ Completions by Independent Facts                             │
\subsection{Completions by Independent Facts}
As we motivated above, we want to use our construction of tuple-independent
PDBs to obtain a completion of a given probabilistic database.\par

We let $\dd=(\Omega,\fr A,P)$ be a PDB of schema $\tau$ and universe $\UU$;
this is the PDB we shall complete. We assume that the occurrences of new facts
are independent in the completion of $\dd$. For the moment, we leave more
sophisticated completions open as future work (see \cref{sec:disc}).
Since our constructions of tuple-independent PDBs always yield countable PDBs
anyway, for convenience we assume that the universe $\UU$ is countable. Then
$\fr A=2^\Omega$.  Let $F(\dd)$ be the set of facts that appear in the
instances of $\Omega$.

\begin{definition}
A completion $\dd'$ of $\dd$ is called \emph{completion by independent facts}
(\emph{independent-fact completion}) if in $\dd'$, all sequences $(\cal
E_F)_{F\in\cal F}$ are independent for collections $\cal F$ of disjoint sets of
facts from $\Facts-F(\dd)$.
\end{definition}

Note that the above definition itself can be easily formulated for arbitrary
(even uncountable) original PDBs. As in \cref{lem:oldtinotion}, in the
countable case, the independence condition is equivalent to the independence of
$(\cal E_f)_{f\in\Facts-F(\dd)}$ in $\dd'$.\par

We remark that especially, we do not allow any facts from $\Facts-F(\dd)$ to
have probability $1$ (otherwise $P'(\Omega)=0$).\par

For the following, we assume that $\Omega$ (the sample space of $\dd$) is
closed under subsets and union. This restriction will be revisited later.

\begin{theorem}\label{thm:indfactcompl}
Let $(p_f)_{f\in\Facts-F(\dd)}$ be a sequence of numbers $p_f\in[0,1)$ such
that $\sum_f p_f <\infty$. Then we can construct an independent-fact completion
$\dd'$ of $\dd$ with the property that $P'(\cal E_f)=p_f$ for all
$f\in\Facts-F(\dd)$ where $P'$ is the probability measure of\/ $\dd'$.
\end{theorem}

\begin{proof}[Proof of \cref{thm:indfactcompl}]
Let $\cc$ be the \ti{} PDB with sample space $\{D\subseteq\Facts-F(\dd) \colon
D\text{ finite}\}$ that is constructed as described in
\cref{ssec:construction}. Let $P_1$ denote the probability measure of $\cc$.
We now define a PDB $\dd'$ with sample space $\Omega' = \dbinst$: every
instance of $\dd'$ is a unique disjoint union $D' = D\uplus C$ with
$D\in\Omega$ and an instance $C$ of $\cc$. We set \[P'(\{D'\}) = P(\{D\})\cdot
P_1(\{C\})\text.\] This yields a probability distribution (in fact, a product
distribution).  For original instances $D\in\Omega$, we have

\begin{equation}\label{eq:composition}
P'(\{D\}) = P(\{D\})\cdot P_1(\{\emptyset\})
\end{equation}

and $P_1(\{\emptyset\}) > 0$ since $\cc$ contains no facts of probability $1$.
By distributivity, an analogous version of \eqref{eq:composition} holds for
\emph{sets} of original instances. Hence, for every $D\in\Omega$,

\begin{equation*}
  P'\big(\{D\}\bigunder\Omega\big)
= \frac{P'(\{D\}\cap\Omega)}{P'(\Omega)}
= \frac{P(\{D\})\cdot P_1(\{\emptyset\})}
        {P_1(\{\emptyset\})}
= P\big(\{D\}\big)\text.\qedhere
\end{equation*}

\end{proof}

Let us now, as previously announced, review the assumption that $\dd$ be closed
under (countable) union and subsets of instances (this was used for the easy
decomposed representation of new instances). Suppose, we want to complete
$\dd_0 = (\Omega_0,2^{\Omega_0},P_0)$ where $\Omega_0$ is a proper subset of
$\big\{D\subseteq F(\dd_0)\colon D\text{ finite}\big\}$.  We can add the
\enquote{missing} instances in the following way: fix some $c\in[0,1]$ and
define a PDB $\dd = (\Omega,2^\Omega,P)$ with $\Omega$ being the set of
(finite) subsets of $F(\dd_0)$ such that $P(\{D\}) = cP_0(\{D\})$ whenever
$D\in\Omega_0$ and $P(\Omega-\Omega_0) = 1-c$ (by specifying probabilities for
the instances of $\Omega-\Omega_0$ with a total mass of $1-c$).

\begin{remark}
Note that this extension of $\dd_0$ to $\dd$ is reasonable, if $\dd_0$ is
finite but harder to motivate (although technically possible) if $\dd_0$ is
itself countably infinite and infinitely many facts are \enquote{missing}.  On
the other hand note that countable PDBs already fulfill the required closure
properties if they are tuple-independent, in which case no such extension is
required.
\end{remark}

Now execute the construction from the proof above for the resulting PDB and
observe that the completion condition is satisfied.

\begin{equation*}
  P'\big(\{D\}\under\Omega_0\big)
= \frac{P'(\{D\})}{P'(\{\Omega_0\})}
= \frac{c\cdot P_0(\{D\})) \cdot P_1(\{\emptyset\})}
        {c\cdot P_1(\{\emptyset\})}
= P_0\big(\{D\}\big)
\end{equation*}

for every $D\in\Omega_0$.\par

Analogously, $\dd$ might be, for example, augmented by finitely many,
arbitrarily correlated instances of arbitrary probability mass before carrying
out the completion.\par\bigskip

\cref{thm:indfactcompl} assures the existence of an infinite open-world
approach for countable PDBs and establishes in some sense a generalization of
the model of Ceylan et al. \cite{Ceylan+2016}. If the given universe $\UU$ is
finite, we can directly obtain their framework. In this case we only need to
specify probabilities for finitely many new facts. In \cite{Ceylan+2016}, the
authors construct a collection of finite PDBs that contains all the completions
of the original PDB by probabilities up to some (reasonably small) upper bound
$\lambda$. The generalization of this idea is also achievable in our setting:
instead of a fixed upper bound $\lambda$, fact probabilities could be bounded
by the summands of a fixed convergent series.\par

\begin{example}\label{exa:completion}
We want to close this section with a small, abstract example. Supposed $\UU =
\{\tt A,\tt B,\tt C,\tt D\}\cup\NN$ and let $\tau = \{R\}$ consist of a single,
binary relation symbol. Consider the following finite \ti{} PDB
$\dd=(\Omega,2^\Omega,P)$ where the last column displays the probabilities
$P(\cal E_f)$.

\begin{center}
\begin{tabular}{l r | c}
\multicolumn{2}{c |}{$R$} & $P$ \\\hline
$\tt A$ & $1$ & $0.8$\\
$\tt B$ & $1$ & $0.4$\\
$\tt B$ & $2$ & $0.5$\\
$\tt C$ & $3$ & $0.9$
\end{tabular}
\end{center}

Additionally, assume $R$ is supposed to be a relation between $\{\tt A,\tt B,
\tt C, \tt D\}$ and $\NN$ (this is for instance achievable by excluding facts
of the wrong shape from $\Facts$). The usual closed-world interpretation of the
tabular representation above would be a PDB over the universe $\UU'=\{\tt A,
\tt B,\tt C,1,2,3\}$ and, for example, the probability that two facts of the
shape $R(\tt A,i)$ are occurring would be $0$. Also, the object $\tt D$ would
not occur whatsoever.\par

Instead, we want to apply the open-world assumption to $\dd$ by assuming that
the probability of any unspecified tuple $(x,i)$ to belong to $R$ is given by
$2^{-1}$ (i.\,e. there are up to 4 facts $f$ with probability $2^{-i}$ for
every $i$). Obviously, the sum of all fact probabilities converges. Hence,
these probabilities induce an independent-fact completion $\dd'$ of $\dd$. In
particular, in $\dd'$, all finite Boolean combinations of (occurrences of)
distinct facts have probability $> 0$.
\end{example}

% ╒═════════╤═════════════════════════════════════════════════════════════════╕
% │ Section │ A Naïve Approximation of Query Evaluation                       │
% ╘═════════╧═════════════════════════════════════════════════════════════════╛
\section{A Na\"{i}ve Approximation of Query Evaluation}
\label{sec:query-evaluation}%
In this section, we investigate the problem of query evaluation. Its purpose is
to demonstrate, that query evaluation for infinite PDBs is not out of reach
from an algorithmic perspective. This may serve as a stepping stone in further
more thorough examination of the subject.\par

We consider the following setting. Let $\UU$ be some countable universe and
$\tau$ be a database schema. We also assume that $\UU$ is computable, for
example $\UU=\Sigma^*$ for some finite alphabet $\Sigma$, so that an algorithm
can generate all facts $f\in\Facts$.  Given is an infinite \ti{} PDB
$\dd=(\Omega, 2^\Omega,P)$ over $\tau$ and $\UU$ and a \emph{first order} query
$Q(\vec x)$ with free variables $\vec x$, and we want to compute $P(Q) =
P(\{D\in\Omega\colon D\models Q\})$. As we have to deal with an infinite PDB,
we will not exactly evaluate queries but instead discuss, how query results can
be approximated up to an arbitrarily small error. Our focus remains on Boolean
queries $Q$ for the moment. We will hint on how to process non-Boolean queries
later.\par

Let $F(\dd)$ be the set of facts appearing among the instances of our PDB $\dd$
and let $p_f\coloneqq P(\cal E_f)$. We make two assumptions concerning our
access to the probability measure of $\dd$:

\begin{enumerate}
\item[(i)] the expected size $\Ex(S_\dd)=\sum_{f\in F(\dd)} p_f$ of $\dd$ is
       known and
\item[(ii)] given $f$, we have oracle access to $p_f$.
\end{enumerate}

Note that these two assumptions are, for example, easily achievable if we
obtained $\dd$ by completing a finite \ti{} PDB as described in
\cref{sec:completions}.\par

\begin{proposition}\label{pro:app}
Let $0<\epsilon<\frac{1}{2}$. Then there exists an algorithm that, given a
Boolean query $Q\in\FO[\tau,\UU]$ and access to a tuple-independent PDB
$\dd\in\dbinst$ (via (i),(ii)), computes an \emph{additive} approximation $p$
of $P(Q)$ with error guarantee $\epsilon$, that is,

\begin{equation*}
  P(Q)-\epsilon\quad\underset{\mathclap{(a)}}{\leq}\quad p
                \quad\underset{\mathclap{(b)}}{\leq}\quad P(Q)+\epsilon\text.
\end{equation*}

\end{proposition}

\begin{proofsketch}
We will omit some technical details of the proof in this presentation. They can
be found in detail in the \hyperref[euler1proof]{appendix}.\par

Let $F(\dd) = \{f_1,f_2,\dots\}$ and let $p_i \coloneqq p_{f_i}$. Choose $n$
large enough such that for all $i>n$ we have $p_i\leq\sfrac{1}{2}$ and
$e^{\alpha_n}\le 1+\epsilon$ and $e^{-\alpha_n}\ge 1-\epsilon$. This is
possible because $\alpha_n\to0$ as $n$ approaches $\infty$ and the function
$e^x$ is continuous at $0$.  Also, an appropriate $n$ can be found
algorithmically by \emph{systematically} listing facts until the remaining
probability mass is small enough.\par

Let $r$ be the quantifier rank of the input query $Q$ (that is, the maximum
nesting depth of quantifiers), and let $s$ be the number of constants from
$\UU$ appearing in $Q$. Let $\Omega_n = 2^{\{f_1,\dots,f_n\}}$. As always, we
denote the complement of an event $\cal E$ by $\ol{\cal E}$. Note that every
instance $D$ of $\Omega_n$ is $r$-equivalent (that is, equivalent for Boolean
queries up to quantifier rank $r$) to some finite structure of size $O(n+r+s)$.
Hence, $P(Q\under\Omega_n)$ can be computed by a traditional closed-world query
evaluation algorithm for finite tuple-independent PDBs. We let $p$ be the
output of this evaluation and return $p$ as our approximate answer.\par

\begin{figure}[H]
	\centering
	\begin{tikzpicture}
		\begin{scope}
			\clip[draw] (0,0) ellipse (3cm and 1cm);
			\draw[pattern=north east lines,pattern color=gray]
				(-.75,-.2) ellipse (1.6cm and .6cm);
			\draw (-1.5,-2) to (-1,3);
		\end{scope}
		\node at (-1.75,1.15) {$\Omega_n$};
		\node at (-.65,1.35) {$\ol{\Omega_n}$};
		\node at (-.75,-.2)
			{\contour{white}{$\{D\in\Omega\colon D\models Q\}$}};
	\end{tikzpicture}
	\caption{Our approximate answer is the probability of $D$ satisfying
	$Q$ conditioned on $\Omega_n$ (in the image, the fraction of the left
	side that is shaded). We use rough bounds for the remaining probability
	mass to derive our approximation guarantee.}
	\label{fig:naive}
\end{figure}
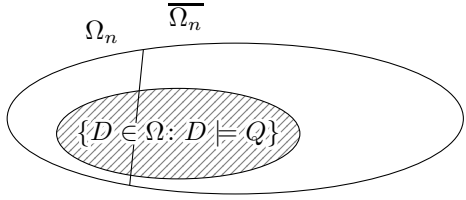

We will now establish the bounds on the error of this approximation. For an
illustration of the situation, see \cref{fig:naive} above.\par

First, we can show that 

\begin{equation}\tag{$*$}\label{eq:euler1}%
P(\Omega_n)=\smashoperator{\prod_{f\notin\{f_1,\dots,f_n\}}}\;\; (1-p_f) \geq 
e^{-\alpha_n}
\end{equation}

(proven in the \hyperref[euler1proof]{appendix}). From this inequality we infer

\begin{equation*}\label{eq:app}
P(Q) = {\underbrace{P(\Omega_n)}_
                    {\leq 1}}
       \cdot p +
       {\underbrace{P(\ol{\Omega_n})}_
                    {\mathclap{\leq 1-e^{-\alpha_n}\leq\epsilon}}}
       \cdot{\underbrace{P(Q\under\ol{\Omega_n})}_
                         {\leq 1}}
\end{equation*}

and therefore immediately $p\geq P(Q)-\epsilon$, showing
\hyperref[eq:app]{(a)}.\par

Towards \hyperref[eq:app]{(b)}, we have 

\begin{equation*}
P(Q) = {\underbrace{P(\Omega_n)}_{\geq e^{-\alpha_n}}}
       \cdot p +
       {\underbrace{P(\ol{\Omega_n})\cdot P(Q\under\ol{\Omega_n})}_
                     {\geq 0}}
\end{equation*}

and hence $p \leq e^{\alpha_n} P(Q)\leq(1+\epsilon)P(Q)\leq P(Q)+\epsilon$.
\qedhere
\end{proofsketch}\par

As we noted before, the additive approximation of \cref{pro:app} can
be extended to allow the evaluation of \FO-queries with free variables. Here we
use the relaxed version of query semantics that was introduced in
\cref{ssec:querysem} where we are only interested in marginal
probabilities of different tuples belonging to the result. These probabilities
can be approximated in the following way: suppose $Q=Q(\vec x)$ where $\vec
x=(x_1,\dots,x_k)$ are the free variables of the $\FO$-formula $Q$. From $Q$ we
can obtain $\lvert\adom(\Omega_n)\rvert^k$ many sentences $Q(\vec a)$ by
plugging in all the possible valuations $\vec a$ of $\vec x$ from $\adom(
\Omega_n)^k$ as constants. The probability of $\vec a$ to belong to the output
of the query $Q$ is equal to the probability of the sentence $Q(\vec a)$ being
satisfied in our PDB. With the procedure above, this probability can be
approximated up to an additive error of $\epsilon$. Note that this
approximation only contains facts from $\Omega_n$.\par\bigskip

The following proposition shows that there is no hope to replace the additive
approximation guarantee of \cref{pro:app} by a multiplicative one
(which is more common in approximation algorithms). We cannot even do this for
a very simple fixed conjunctive query.  Let $\Sigma$ be a finite alphabet and
let $\tau$ be a database schema. We say that a Turing machine $M$
\emph{represents} a \ti{} PDB over $\Sigma,\tau$ of \emph{weight} $w$ if it
computes a function $p_M \colon \Facts[\tau,\Sigma^*] \to \ds Q$ such that
$\sum_{f\in\Facts[\tau,\Sigma^*]}p_M(f) = w$. The PDB $\dd_M$ represented by
$M$ is the tuple-independent PDB with universe $\UU = \Sigma^*$, schema $\tau$
and fact probabilities $p_M(f)$. Note that if we have a Turing machine $M$
representing a PDB $\dd_M$ in this sense then the two assumptions (i), (ii) of
\cref{pro:app} are again satisfied with $p_f\coloneqq p_M(f)$.\par

\begin{proposition}\label{pro:multiplicative}
Let $\Sigma=\{0,1\}$ and $\tau=\{R,S\}$ for a unary relation symbols $R, S$.
Let $Q$ be the Boolean query $\exists x\colon R(x)$ in $\FO[\tau,\UU]$.
Furthermore, let $c\geq 1$. There is no algorithm $A$ that, given a Turing
machine $M$ representing a tuple-independent PDB over $\Sigma,\tau$ of weight
$1$, computes a number $p$ such that

\begin{equation*}
      {\sfrac{1}{c}}\cdot\Pr\nolimits_{D\sim\dd_M}(D\models Q)
\leq p
\leq c\cdot\Pr\nolimits_{D\sim\dd_M}(D\models Q)\text.
\end{equation*}

\end{proposition}

The detailed proof can be found in the \hyperref[app:62]{appendix}.\par

Let us close this section with some remarks regarding complexity issues of the
previously described approximation procedure. Basically, its run-time is given
by the run-time of the finite evaluation algorithm when applied to a PDB with a
universe of size $n$. In the proof of \cref{pro:app},
$n=n(\epsilon)$ was the number of facts that needed to be taken into
consideration in order to obtain the error guarantee $\epsilon$ and is
basically determined by the rate of convergence of the series of fact
probabilities. The way we produced $n$ systematically ensures its existence. In
the best case, the facts $f_1,f_2,\dots$ are enumerated by decreasing
probability. For a geometric series of fact probabilities for example,
$n=\Omega\big(\log\big(\sfrac{1}{1-\epsilon}\big)\big)$. It is worth noting,
though, that series in general may converge arbitrarily slowly
\cite[pp.~310-311]{Knopp1996}. For the moment, we leave it at that and refer to
future work for a more thorough examination of the complexity of query
evaluation in infinite PDBs.

% ╒═════════╤═════════════════════════════════════════════════════════════════╕
% │ Section │ Conclusions                                                     │
% ╘═════════╧═════════════════════════════════════════════════════════════════╛
\section{Conclusions}\label{sec:disc}
In this work, we proposed a framework for probabilistic databases that extends
the standard finite notion, which dominated theoretical research on
probabilistic databases so far. Our model provides a theoretical foundation for
several practical systems allowing for values from infinite domains (albeit
still in a restricted way) and opens avenues to new, even more flexible
systems.\par

We discussed independence assumptions in infinite PDBs, most notably the simple
model of tuple-independence. We showed how to construct countable
tuple-independent PDBs realizing any given sequence of fact probabilities,
provided the sum of these fact probabilities converges, and we also proved that
the convergence condition is necessary. An important application of this result
is that it allows us to complete PDBs to cover all potential instances (with
respect to the underlying domain). We also gave a construction of countable
block-independent disjoint probabilistic databases with given fact
probabilities.  Although we did not focus on algorithmic questions, the
aforementioned completions provide the mathematical background for applying
open-world semantics to (classically closed-world) finite PDBs.\par

In general, we expect query evaluation (even approximate) to be difficult in
infinite PDBs. The way to approach it may be to combine classical database
techniques with probabilistic inference techniques from AI, as they are used
for relational languages like BLOG~\cite{Milch+2005},
ProbLog~\cite{DeRaedt+2016}, and Markov logic~\cite{Singla+2007}. However, the
underlying inference problems have a high computational complexity, and
algorithms are mostly heuristic, so we see little hope for obtaining algorithms
tractable in the worst case.  As we showed, the class of tuple-independent PDBs
with respect to some countable universe and schema is not powerful enough to
capture all possible probability spaces, even when extended with $\FO$-views. A
more detailed investigation of the exact boundaries of expressivity, as well as
the corresponding considerations for \bid{} PDBs are still pending. We think
that concise and powerful representation systems for infinite PDBs are of
general interest, even if they might turn out to be only possible as
approximations (in some sense) of arbitrary PDBs. For whatever models or
systems come forth, the consecutive goal is to have efficient (approximation)
algorithms that perform query evaluation, perhaps among other database specific
operations on our probability spaces.\par

Our technical results are, at least implicitly, mostly about countable PDBs. It
would be nice to extend these results, for example the construction of
tuple-independent PDBs to uncountable PDBs in meaningful way. But in fact, even
more basic questions regarding the construction of suitable $\sigma$-algebras
and probability spaces and the measurability of queries and views need a
thorough investigation for uncountable PDBs. Of course, algorithmic
tractability becomes even more challenging in the uncountable.

% ╒═════════╤═════════════════════════════════════════════════════════════════╕
% │ Section │ Acknowledgments                                                 │
% ╘═════════╧═════════════════════════════════════════════════════════════════╛
\section*{Acknowledgments}
We are very grateful to Dan Suciu, for the time he spent reading a manuscript
version of this paper and his valuable comments. In particular, he proposed the
condition $\cal E_F \in\mathfrak A$ for measurable $F$ in
\cref{def:pdb} as a strengthening of our original condition $\cal
E_f\in\fr A$ and made various other proposals improving our treatment of
uncountable PDBs. We also highly appreciate the work of the reviewers for their
suggestions and their criticism and for guiding our attention to some relevant
related works and towards reconsidering some technical issues.  Additionally,
we would like to thank Steffen van Bergerem for helpful remarks.\par

This work is supported by the \href{http://www.dfg.de/}{German Research Council
(DFG)} \href{http://gepris.dfg.de/gepris/projekt/282652900}{Research Training
Group 2236 UnRAVeL}.

% ╔═══════════════════════════════════════════════════════════════════════════╗
% ║                           B I B L I O G R A P H Y                         ║
% ╚═══════════════════════════════════════════════════════════════════════════╝
\phantomsection
\clearpage

% ╔═══════════════════════════════════════════════════════════════════════════╗
% ║                           A P P E N D I X                           *app* ║
% ╚═══════════════════════════════════════════════════════════════════════════╝
\clearpage\appendix

% ╒═════════╤═════════════════════════════════════════════════════════════════╕
% │ Section │ Appendix: Omitted Proofs                                        │
% ╘═════════╧═════════════════════════════════════════════════════════════════╛
\section*{Appendix: Omitted Proofs}\label{app:omittedproofs}
\addcontentsline{toc}{section}{Appendix: Omitted Proofs}

% ┌────────────┬──────────────────────────────────────────────────────────────┐
% │ Subsection │ Proof of Lemma 2.3                                           │
\subsection*{Proof of \cref{lem:tech1}}\label{app:23}
\addcontentsline{toc}{subsection}{Proof of Lemma 2.3}

\begin{apxclaim}
Let $(a_i)_{i\in I}$ be a countably infinite sequence of real numbers such that
the series $\sum_i a_i$ is absolutely convergent. Then it holds that $\prod_i
\big(1+a_i\big) = \sum_{J\subseteq I\text{\normalfont, finite}} \prod_{j\in J}
a_j$ and both sides of this equation are absolutely convergent.
\end{apxclaim}

\begin{proof}[Proof \cite{Simmons2013}]
Wlog., we take $I=\NN$. Since $\sum_i a_i$ is absolutely convergent, so is
$\prod_i (1+a_i)$ by \cref{fac:knopp}. In particular, $\prod_i (1+a_i)$ is
convergent. Then,

\begin{align*}
   \prod_{i\geq 1}\big(1+\lvert a_i\rvert \big)
=& \lim_{n\to\infty} \sprod_{i=1}^n \big(1+\lvert a_i\rvert\big)\\
=& \lim_{n\to\infty}\mkern6mu
   \smashoperator{\sum_{J\subseteq\{1,\dots,n\}}}\,
   {\textstyle\sprod_{j\in J}} \lvert a_j\rvert\\
=& \lim_{n\to\infty}\mkern6mu
   \smashoperator{\sum_{J\subseteq\{1,\dots,n\}}}\,
   \left\lvert{\textstyle\sprod_{j\in J}} a_j\right\rvert\\
=& \smashoperator[l]{\sum_{\substack{J\subseteq\NN\\
	\text{\normalfont finite}}}}
	\mkern4mu\sprod_{j\in J} \lvert a_j\rvert \text.
\end{align*}

The last notation used in the last equation is motivated by the (immediate)
absolute convergence of the series. Exactly the same calculation, omitting
$\lvert \cdot\rvert$, shows 
\begin{equation*}
	\prod_{i\geq 1}(1+a_i)
	= \sum_{\substack{J\subseteq \ds N\\\text{finite}}} 
	\sprod_{j\in J} a_j\text.
\end{equation*}
\end{proof}

% ┌────────────┬──────────────────────────────────────────────────────────────┐
% │ Subsection │ Proof of Lemma 4.12                                          │
\subsection*{Proof of \cref{lem:oldbidnotion}}\label{app:412}
\addcontentsline{toc}{subsection}{Proof of Lemma 4.12}%
Let $\UU$ be some universe and $\tau$ be some database scheme.

\begin{apxclaim}
Let $\dd = (\Omega,\fr A,P)$ be a countable PDB over $\tau$ and\/ $\UU$ and let
$\cal B$ be a partition of $\Facts$ such that for each $B\in\cal B$, the events
\/ $\cal E_{B_1}$ and\/ $\cal E_{B_2}$ are mutually exclusive for any disjoint
$B_1,B_2 \subseteq B$. Then the following conditions are equivalent:

\begin{enumerate}
\item[(2)] The sequences $(\cal E_{B_i'})_{1\leq i\leq k}$ are independent for
	all $k\in\NN$ and all $B_1',\dots,B_k'$ being measurable subsets from
	mutually different blocks.\label{eq:2prime}
\item[(2')] The sequences $(\cal E_f)_{f\in F}$ are independent in $\dd$ for
	every set $F$ of facts such that $F$ contains at most one fact from
	each block.
\end{enumerate}

\end{apxclaim}

\begin{proof}
The implication ($\Rightarrow$) holds by definition.\par

For the other direction let $\dd=(\Omega,\fr A, P)$ be a countable \bid{} PDB.
Let $F_\omega$ be the (countable) set of facts $f$ with $P(\cal E_f)>0$ and let
$\cal B_\omega$ be the (countable) set of blocks belonging to
$F_\omega$.\par\medskip

We prove the following intermediate claim: let $\cal B_0$ be the partition of
$\Facts$ with $\cal B_0 = \{B\colon B = B'\cap F_\omega\text{ for some
}B'\in\cal B_\omega\} \cup \{B_0\}$ where $B_0 = \{f\colon P(\cal E_f)=0\}$
(note that $\cal B_0$ and $\cal B$ only differ in null sets). Then the
following holds.

\begin{center}
If $\dd$ is \bid{} wrt. $\cal B$, then it is \bid{} wrt. $\cal B_0$.
\end{center}

This is easy to see. First consider condition \eqref{itm:bidexcl} from the
definition of \bid{} PDBs. Let $B_1'$ and $B_2'$ be disjoint measurable sets
contained in the same block $B$ of $\cal B_0$. If they are from $B_0$, they are
trivially exclusive, because in this case both of them have measure zero. Note
that we used countability here. Otherwise, they were disjoint measurable
subsets of the same original block from $\cal B$ and hence also exclusive.\par

Now consider \eqref{itm:bidexcl} from the \bid{} definition. We claim that
events $\cal E_{B_i'}$ are independent for all finite collections $(B_i')$ of
measurable sets from different blocks. If one of the $B_i'$ is contained in
$B_0$, its measure is $0$ (again by countability of $\dd$) and the claimed
independence immediate. Otherwise, all $B_i'$ were measurable subsets of
original blocks from $\cal B_\omega$. Hence, they are independent.\par\medskip

Due to this observation, we may assume that the blocks of $\dd$ are given by
$\cal B_0$ as defined above. In the following, we let thus be $\cal B = \cal
B_0$ and use $\cal B_\omega$ in the same meaning as defined above for blocks of
$\cal B$ (i.\,e. the \emph{new} $\cal B_\omega$ is obtained by restricting the
\emph{old} blocks of $\cal B_\omega$ to $F_\omega$).\par

We proceed to show \hyperref[eq:2prime]{(2)} for the blocks from the partition
described above. Let $B_1',\dots, B_k'$ be a sequence of measurable subsets of
distinct blocks from $\cal B_\omega$ (we may restrict our consideration to
$\cal B_\omega$ since all events $\cal E_f$ with $f$ belonging to $B_0$ are
null sets). Let $\cal E_i$ denote $\cal E_{B_i'}$.

\begingroup\allowdisplaybreaks
\begin{align*}
     P\vleft(\smashoperator[r]{\bigcap_{1\leq i\leq k}}\;\cal E_i\vright)
={}& P\vleft(\bigcap_{1\leq i\leq k} \bigcup_{b\in B_i'} \cal E_b\vright)\\
={}& P\vleft(\smashoperator[r]{\bigcup_{\substack{(b_1,\dots,b_k)\\
                                        \in B_1'\times\dots\times B_k'}}}
             \;\cal E_{b_1} \cap \dots \cap \cal E_{b_k}\vright)\\
={}& \smashoperator{\sum_{\substack{(b_1,\dots,b_k)\\
                             \in B_1'\times\dots\times B_k'}}}
     \;P\big(\cal E_{b_1}\cap\dots\cap \cal E_{b_k}\big)\\
={}& \smashoperator{\sum_{b_1\in B_1'}}\;
     \dots\;
     \smashoperator[l]{\sum_{b_k\in B_k'}}\;
     \smashoperator[r]{\prod_{1\leq i\leq k}}\;P\big(\cal E_{b_i}\big)\\
={}& \smashoperator[l]{\prod_{1\leq i\leq k}}
     \sum_{b\in B_i'} P\big(\cal E_{b}\big)
={}  \smashoperator{\prod_{1\leq i\leq k}}\;P\big(\cal E_i\big)\text.\qedhere
\end{align*}\endgroup

\end{proof}

% ┌────────────┬──────────────────────────────────────────────────────────────┐
% │ Subsection │ Proof of Proposition 4.13                                    │
\subsection*{Proof of \cref{pro:bidconstruction}}\label{app:413}
\addcontentsline{toc}{subsection}{Proof of Proposition 4.13}%
Let $\UU$ be some universe and $\tau$ be some database scheme.

\begin{apxclaim}
Let $\cal B$ be a partition of $\Facts$ into blocks and for every block
$B\in\cal B$ let $(p_f^B)_{f\in B}$ such that $p_f^B\in[0,1]$ and $\sum_{f\in
B} p_f^B\leq 1$. If

\begin{equation*}
\sum_{f\in F} p_f^{B(f)} < \infty
\end{equation*}

for all finite subsets $F\subseteq\Facts$ and with $B(f)$ being the block
containing $f\in F$, then we can construct a (countable) \bid{} PDB $\dd =
(\Omega,\fr A,P)$ (wrt. $\cal B$) with the property that $P(\cal E_f) =
p_f^{B(f)}$.
\end{apxclaim}

\begin{proof}
Let $\Omega$ be the set of finite subsets of $\Facts$ and $\fr A\coloneqq
2^\Omega$. Just like in the proof of \cref{lem:oldbidnotion}, the set
$F_\omega$ of facts with positive marginal probability is countable. Again, we
may suppose that all impossible facts are bundled into a dummy block $B_0$ such
that all the (countably many) remaining blocks $\cal B_\omega$ are countable
and cover exactly $F_\omega$.\par%
We call an instance \emph{good}, if it contains at most one fact from every
block $B$. Otherwise, it is called \emph{bad}. Let $\Omega_+$ ($\Omega_-$) be
the set of good (bad) instances. We define a mapping $\beta\colon\cal
B\times\Omega_+\to\Facts\cup\{\bot\}$ that, given a block $B$ and a good
instance $D$ returns the (unique, if existent) fact $f$ from $D$ lying in $B$
and returns \enquote{$\bot$} if $D$ does not contain a fact from $B$.

\begin{equation*}
\beta(B,D)\coloneqq
\begin{dcases}
f    & \text{if }D\cap B=\{f\}\text,\\
\bot & \text{if }D\cap B=\emptyset\text.
\end{dcases}
\end{equation*}

For $D\in\dbinst$, we set

\begin{equation*}
P\big(\{D\}\big) \coloneqq
\begin{dcases}
\sprod_{B\in\cal B} p_{\beta(B,D)}^B & \text{if }D\in\bf \Omega_+\text,\\
0                                    & \text{if }D\in\bf \Omega_-\text.
\end{dcases}
\end{equation*}

where $\smash{p_\bot^B\coloneqq 1-\sum_{f\in B} p_f^B\in[0,1]}$ is the
\emph{remainder mass} of block $B\in\cal B$. Letting $\bf D_\omega$ denote the
set of finite subsets of $F_\omega$, we complete the definition of our
probability space by setting $P(A) = \sum_{D\in A\cap \bf D_\omega} P(\{D\})$
for every further element $A$ of $\fr A$. For $P$ to be a probability measure,
we will show $P(\Omega)= P(\bf D_\omega) = 1$. From now on, the reasoning will
proceed analogously to the various proofs regarding the tuple-independence
construction of \cref{ssec:construction}. Since the notions are however
slightly more involved, we present the full proof below.\par

For an instance $D$ let $\cal B_D \coloneqq \{B\in\cal B\colon B\cap D\neq
\emptyset\}$ be the set of blocks $B\in\cal B$ such that $D$ contains a fact
from $B$. Let $\bf D_\omega^+ \coloneqq \bf D_\omega \cap \Omega_+$. Observe

\begin{align*}
   P(\Omega)
&= P(\bf D_\omega^+)\\
&= \sum_{D\in\bf D_\omega^+} P(\{D\})\\
&= \sum_{\substack{\cal B'\subseteq\cal B_\omega\\\text{\normalfont finite}}}
   \sum_{\substack{D\in\bf D_\omega^+\\\cal B_D = \cal B'}}
   \sprod_{B\in\cal B'} p_{\beta(B,D)}^B
   \sprod_{B\in\cal B_\omega\setminus\cal B'} p_\bot^B\\
&= \sum_{\substack{\cal B'\subseteq\cal B_\omega\\\text{\normalfont finite}}}
   \sprod_{B\in\cal B_\omega\setminus\cal B'} p_\bot^B
   \sum_{\substack{D\in\bf D_\omega^+\\\cal B_D=\cal B'}}
   \sprod_{B\in\cal B'} p_{\beta(B,D)}^B.
\end{align*}

Note that we may omit the dummy block from the calculations, since it is only
present via $p_\bot^B=1$.  Consider the inner sum and suppose $\cal B' =
\{B_1,\dots,B_k\}$. Then

\begin{align}
   \sum_{\substack{D\in\bf D_\omega^+\\\cal B_D=\cal B'}}
   \sprod_{B\in\cal B'} p_{\beta(B,D)}^B
&= \smashoperator{\sum_{f_1\in B_1}}\mkern16mu
   \dots\mkern12mu
   \smashoperator{\sum_{f_k\in B_k}}
   \mkern4mu p_{f_1}^{B_1}\dots\mkern2mu p_{f_k}^{B_k}\notag\\
&= \smash[t]{\vleft(\sum_{f_1\in B_1}p_{f_1}^{B_1}\vright)}
   \mkern4mu\dots
   \smash[t]{\vleft(\sum_{f_k\in B_k}p_{f_k}^{B_k}\vright)}\notag\\
&= \prod_{B\in\cal B'}\big(1-p_{\bot}^B\big)\label{eq:appendixsum1}\text.
\end{align}

In order to keep the similarity to the proof of \cref{ssec:construction}, we
let $p_\top^B$ denote $1-p_\bot^B$. Then continuing the above calculation and
proceeding analogously to the proof of \cref{lem:tech2} (which makes in
particular uses \cref{lem:tech1}, we have 

\bgroup\allowdisplaybreaks%
\begin{align*}
P(\Omega)
&= \sum_{\substack{\cal B'\subseteq\cal B_\omega\\\text{\normalfont finite}}}
   \sprod_{B\in\cal B_\omega\setminus\cal B'} p_\bot^B
   \sum_{\substack{D\in\bf D_\omega^+\\\cal B_D=\cal B'}}
   \sprod_{B\in\cal B'} p_{\beta(B,D)}^B\\
&= \sum_{\substack{\cal B'\subseteq\cal B_\omega\\\text{\normalfont finite}}}
   \sprod_{B\in\cal B_\omega\setminus\cal B'} p_\bot^B
   \sprod_{B\in\cal B'}\big(1-p_{\bot}^B\big)\\
&= \sum_{\substack{\cal B'\subseteq\cal B_\omega\\\text{\normalfont finite}}}
   \sprod_{B\in\cal B'}p_\top^B
   \sprod_{B\in\cal B_\omega\setminus\cal B'} \big(1-p_\top^B\big)\\
&= \sum_{\substack{\cal B'\subseteq\cal B_\omega\\\text{\normalfont finite}}}
   \sprod_{B\in\cal B'}p_\top^B
   \sum_{\substack{\cal B''\subseteq \cal B_\omega-\cal B'\\
   \text{\normalfont finite}}}
   \sprod_{B\in\cal B''}\big(-p_\top^B\big)\\
&= \sum_{\substack{\cal B'\subseteq\cal B_\omega\\\text{\normalfont finite}}}
   \sum_{\substack{\cal B_\omega \supseteq \cal B''\supseteq \cal B'\\
   \text{\normalfont finite}}}
   \sprod_{B\in\cal B'}p_\top^B
   \sprod_{B\in\cal B''-\cal B'}\big(-p_\top^B\big)\\
&= \sum_{\substack{\cal B''\subseteq\cal B_\omega\\\text{\normalfont finite}}}
   \sum_{\cal B' \subseteq \cal B''}
   \sprod_{B\in\cal B'}p_\top^B
   \sprod_{B\in\cal B''-\cal B'}\big(-p_\top^B\big)\\
&= 1\text.
\end{align*}
\egroup%

The last step is justified by the same reasoning as in the proof of 
\cref{lem:tech2}: for $\cal B'' = \emptyset$, the inner sum consists only of an
empty product and thus equals $1$; otherwise, the inner sum evaluates to $0$
(which can be seen by fixing some $B''\in\cal B''$ and splitting the inner sum
into two sums---one with $B''\in\cal B'$ and one with $B''\notin\cal B'$;
factoring out $\smash{p_\top^{B''}}$ respectively $\smash{-p_\top^{B''}}$, it
is easy to see that both sums cancel each other out).\par

Now that we have established that $P$ is a probability measure, we still have
to show that $\dd$ is block-independent-disjoint. By \cref{lem:oldtinotion}, it
suffices to show the independence of the events $\cal E_f$ for facts from
different blocks. Let thus $F$ be a finite set of facts from $F_\omega$ such
that $F$ contains at most one fact per block. For all facts $f$, let $B(f)$
denote the block that contains $f$ and $\cal B(F)\coloneqq\{B(f)\colon f\in
F\}$. Let $\smash{p_f \coloneqq p_f^{B(f)}}$.  Let $\Omega_F$ be the set of
\emph{good} instances containing $F$.

\begin{align}
   P\vleft(\bigcap_{f\in F}\cal E_f\vright)
&= \sum_{D\in\Omega_F} P\big(\{D\}\big)\notag\\
&= \sum_{D\in\Omega_F} \sprod_{B\in\cal B_\omega} p_{\beta(B,D)}^B\notag\\
&= \underbrace{\prod_{B\in\cal B_F} p_f^B}_{=\prod_{f\in F}p_f}
   \vleft(\sum_{D\in\Omega_F}
          \sprod_{B\in\cal B_\omega\setminus\cal B_F} p_{\beta(B,D)}^B\vright)
   \label{eq:appendixcapEf}
\end{align}

Like in the proof of \cref{lem:constristi}, we show that the parenthesized
term equals $1$. Note that this sum only ranges over blocks \emph{not} from
$\cal B_F$ (excluding the dummy block since $F\subseteq F_\omega$ and it would
hence appear as a factor $p_\bot^B = 1$, which we omit). Note that the summand
for $D\in\Omega_F$ is equal to the product $\sprod_{B\in\cal B_\omega-\cal B_F}
p_{\beta(B,D')}^B$ where $D'=D-F$ and that subtracting $F$ constitutes a
bijection between $\Omega_F$ and $\Omega_F'\coloneqq\big\{D\in\bf
D_\omega^+\colon D\cap B=\emptyset \text{ for all }B\in\cal B_F\big\}\text.$
Hence,

\begin{equation*}
   \sum_{D\in\Omega_F}
   \sprod_{B\in\cal B_\omega-\cal B_F} p_{\beta(B,D)}^B\\
=  \sum_{D\in\Omega_F'}
   \sprod_{B\in\cal B_\omega-\cal B_F} p_{\beta(B,D)}^B\text.
\end{equation*}

Now suppose $\cal B_F = \{B_1,\dots,B_k\}$ and let $\tilde B_i=B_i\cup\{\bot\}$
($1\leq i\leq k$). Then

\begin{equation*}
   \smashoperator{\prod_{\substack{(\tilde f_1,\dots, \tilde f_k)\\
                    \in \tilde B_1 \times \dots \times \tilde B_k}}}
   \mkern12mu p_{\tilde f_i}^{\tilde B_i}
=  1
\end{equation*}

with an easy calculation like in \eqref{eq:appendixsum1}. Thus,

\begin{align*}
   \sum_{D\in\Omega_F'}
   \sprod_{B\in\cal B_\omega-\cal B_F} p_{\beta(B,D)}^B
=& \sum_{D\in\Omega_F'}
   \sprod_{B\in\cal B_\omega-\cal B_F} p_{\beta(B,D)}^B
   \sprod_{\substack{(\tilde f_1,\dots, \tilde f_k)\\
                    \in \tilde B_1 \times \dots \times \tilde B_k}}
   p_{\tilde f_i}^{\tilde B_i}\\
=& \sum_{D\in\bf D_\omega^+} P\big(\{D\}\big) = 1\text.
\end{align*}

Since $P(\cal E_f) = p_f$ (this is already immediate from the above for
$F=\{f\}$) and continuing at \eqref{eq:appendixcapEf}, we arrive at

\begin{equation*}
	P\vleft(\bigcap_{f\in F}\cal E_f\vright)
=  \prod_{B\in\cal B_F} p_f^B
=  \prod_{f\in F} p_f^{B(f)}\text.\qedhere
\end{equation*}

\end{proof}

% ┌────────────┬──────────────────────────────────────────────────────────────┐
% │ Subsection │ Proof of Claim (*) in Proposition 6.1                        │
\subsection*{Proof of Claim \texorpdfstring{\eqref{eq:euler1}}{(*)} in 
\cref{pro:app}}
\label{euler1proof}\addcontentsline{toc}{subsection}{Proof of Claim (∗) in
	Proposition 6.1}%
\begin{apxclaim}
Let $(p_i)_{i\geq 1}$ be a sequence with $\sum_i p_i < \infty$ and 
$p_i\in\big[0,\sfrac{1}{2}\big)$. Then
\begin{equation*}
\prod_i\big(1-p_i\big) \geq \exp\big({{\sfrac{3}{2}}\textstyle\sum_i p_i}\big)
\end{equation*}
(where $\exp(x)=e^x$ for $x\in\RR$).
\end{apxclaim}
\begin{proof}
For $\lvert x\rvert < 1$, the Taylor series expansion of $\ln(1+x)$ is
\begin{equation*}
\ln(1+x)=\sum_{k\geq 1}\sfrac{(-1)^{k-1}x^k}{k}\text.
\end{equation*}
Hence, with $x\coloneqq -p_i < 0$ and $(-1)^{k-1}(-p_i)^k=-p_i^k$,
\begin{equation*}
1-p_i = \exp\big(-\sum\nolimits_{k\geq 1}\sfrac{p_i^k}{k}\big)\text.
\end{equation*}
Since $p_i\leq\sfrac{1}{2}$,
\begin{equation*}
1\geq \sum_{k\geq 1}p_i^k \geq \sum_{k\geq 1} \sfrac{2}{k+2} p_i^k\text,
\end{equation*}
by multiplying with $-p_i^2/2$, we have
\begin{equation}\label{eq:taylorcutoff}
-\sfrac{p_i^2}{2} \leq - \sum_{k\geq 1}\sfrac{p_i^{k+2}}{k+2}
\end{equation}
and thus
\begin{equation*}
     1-p_i
=    \exp\big(-\sum\nolimits_{k\geq 1}\sfrac{p_i^k}{k}\big)
\mathrel{\overset{\eqref{eq:taylorcutoff}}{\geq}}
     \exp\big(-p_i-\sfrac{p_i^2}{2}-\sfrac{p_i^2}{2}\big)
\geq \exp\big(-\sfrac{3}{2}p_i\big)
\end{equation*}
since $p_i<\sfrac{1}{2}$. The claim follows.
\end{proof}

% ┌────────────┬──────────────────────────────────────────────────────────────┐
% │ Subsection │ Proof of Proposition 6.2                                     │
\subsection*{Proof of \cref{pro:multiplicative}}\label{app:62}
\addcontentsline{toc}{subsection}{Proof of Proposition 6.2}%

\begin{apxclaim}
Let $\Sigma=\{0,1\}$ and $\tau=\{R,S\}$ for a unary relation symbols $R, S$.
Let $Q$ be the Boolean query $\exists x\colon R(x)$ in $\FO[\tau,\UU]$.
Furthermore, let $c\geq 1$. There is no algorithm $A$ that, given a Turing
machine $M$ representing a tuple-independent PDB over $\Sigma,\tau$ of weight
$1$, computes a number $p$ such that

\begin{equation}\label{eq:multapp}
      {\sfrac{1}{c}}\cdot\Pr\nolimits_{D\sim\dd_M}(D\models Q)
\leq p
\leq c\cdot\Pr\nolimits_{D\sim\dd_M}(D\models Q)\text.
\end{equation}

\end{apxclaim}

\begin{proof}
It will be convenient in the proof to identify $\Sigma^*$ with the set $\NN$ of
positive integers (the string $x\in\Sigma^*$ represents the integer with binary
representation $1x$). Moreover, we let $\langle\cdot,\cdot\rangle
\colon\NN^2\to\NN$ be a pairing function (such as $\smash{\langle m,n \rangle =
{\sfrac{1}{2}}\big(x+y-1\big)\big(x+y-2\big)+2}$).\par

For a Turing machine $N$ with input alphabet $\Sigma$, we let $L_N$ be the set
of all $n\in\NN$ accepted by $N$. By Rice's Theorem, the set \EMPTY{} of all
(encodings of) Turing machines $N$ with $L_N = \emptyset$ is undecidable.  For
every $t\in\NN$, let $L_{N,t}$ be the set of all $n\in\NN$ such that $N$
accepts $n$ in at most $t$ steps. Note that $L_{N,t}$ is decidable (even in
polynomial time) and that $L_N = \bigcup_{t\in\NN} L_{N,t}$.\par

We reduce \EMPTY{} to our query evaluation problem. Let $N$ be a Turing machine
with input alphabet $\Sigma$. We construct a Turing machine $M = M(N)$
representing a tuple-independent PDB over $\Sigma,\tau$ of weight $1$ that
works as follows: given a string $f\in(\Sigma\cup\tau\cup\{(,)\})^*$, it checks
whether $f\in\Facts[\tau,\Sigma]$. If this is not the case, it rejects.
Otherwise, $f$ is of the form $R(k)$ or $S(k)$ for some $k\in\NN$. Let $n,t \in
\NN$ such that $k = \langle n,t\rangle$. Then if $f=R(k)$ and $n\in L_{N,t}$
\emph{or} if $f=S(k)$ and $n\notin L_{N,t}$, the machine $M$ returns $p_M(f)
\coloneqq 2^{-k}$. Otherwise, $M$ returns $p_M(f) \coloneqq 0$.  Then
$\sum_{f\in\Facts[\tau,\Sigma*]}p_M(f) = \sum_{k\in\NN} 2^{-k} = 1$.  This
shows that $M$ represents a tuple-independent PDB over $\Sigma$ and $\tau$ of
weight $1$.\par

Moreover, $p_M(R(k)) = 0$ for all $k\in\NN$ if and only if $L_N=\emptyset$.
Thus, $\Pr_{D\sim\dd_M}(D\models Q) = 0$ if and only if $L_N = \emptyset$.\par

Now suppose we have an approximation algorithm $A$ satisfying
\eqref{eq:multapp} for some $c\geq 1$. Then $p = 0$ if and only if
$\Pr_{D\sim\dd_M}(D\models Q) = 0$. Hence we can use $A$ to decide \EMPTY.
\end{proof}


\addcontentsline{toc}{section}{\refname}
\begin{thebibliography}{10}

\bibitem{Abiteboul+2011}
Serge Abiteboul, Tsz-Hong~Hubert Chan, Evgeny Kharlamov, Werner Nutt, and
  Pierre Senellart.
\newblock {Capturing Continuous Data and Answering Aggregate Queries in
  Probabilistic XML}.
\newblock {\em ACM Transactions on Database Systems (TODS)}, 36(4):25:1--25:45,
  2011.

\bibitem{Abiteboul+1995}
Serge Abiteboul, Richard Hull, and Victor Vianu.
\newblock {\em {Foundations of Databases}}.
\newblock Addison-Wesley, Boston, MA, USA, 1st edition, 1995.

\bibitem{Abiteboul+2009}
Serge Abiteboul, Benny Kimelfeld, Yehoshua Sagiv, and Pierre Senellart.
\newblock {On the Expressiveness of Probabilistic XML Models}.
\newblock {\em The VLDB Journal - The International Journal on Very Large Data
  Bases}, 18(5), 2009.

\bibitem{Aggarwal+2009}
Charu~C. Aggarwal and Philip~S. Yu.
\newblock {A Survey of Uncertain Data Algorithms and Applications}.
\newblock {\em IEEE Transactions on Knowledge and Data Engineering (TKDE)},
  21(5):609--623, 2009.

\bibitem{Agrawal+2006}
Parag Agrawal, Omar Benjelloun, Anish~Das Sarma, Chris Hayworth, Shubha Nabar,
  Tomoe Sugihara, and Jennifer Widom.
\newblock {Trio: A System for Data, Uncertainty, and Lineage}.
\newblock In {\em Proceedings of the 32nd International Conference on Very
  Large Data Bases (VLDB '06)}, VLDB '06, pages 1151--1154. VLDB Endowment,
  2006.

\bibitem{Agrawal+2009}
Parag Agrawal and Jennifer Widom.
\newblock {Continuous Uncertainty in Trio}.
\newblock In {\em Proceedings of the 3rd VLDB workshop on Management of
  Uncertain Data (MUD '09)}, pages 17--32, Enschede, The Netherlands, 2009.
  Centre for Telematics and Information Technology (CTIT).

\bibitem{Barany+2017}
Vince B{\'a}r{\'a}ny, Balder Ten~Cate, Benny Kimelfeld, Dan Olteanu, and
  Zografoula Vagena.
\newblock {Declarative Probabilistic Programming with Datalog}.
\newblock {\em ACM Transactions on Database Systems (TODS)}, 42(4):22:1--22:35,
  2017.

\bibitem{Belle2017}
Vaishak Belle.
\newblock {Open-Universe Weighted Model Counting}.
\newblock In {\em Proceedings of the 31st AAAI Conference on Artificial
  Intelligence (AAAI '17)}, pages 3701--3708, Palo Alto, CA, USA, 2017. AAAI
  Press.

\bibitem{Benedikt+2010}
Michael Benedikt, Evgeny Kharlamov, Dan Olteanu, and Pierre Senellart.
\newblock {Probabilistic XML via Markov Chains}.
\newblock {\em Proceedings of the VLDB Endowment}, 3(1--2):770--781, 2010.

\bibitem{Bollobas2001}
B{\'e}la Bollob{\'a}s.
\newblock {\em {Random Graphs}}.
\newblock Cambridge Studies in Advanced Mathematics. Cambridge University
  Press, Cambridge, United Kingdom, 2nd edition, 2001.

\bibitem{Borgwardt+2017}
Stefan Borgwardt, {\.{I}}smail~{\.{I}}lkan Ceylan, and Thomas Lukasiewicz.
\newblock {Ontology-Mediated Queries for Probabilistic Databases}.
\newblock In {\em Proceedings of the 31st {AAAI} Conference on Artificial
  Intelligence (AAAI '17)}, pages 1063--1069, Palo Alto, CA, USA, 2017. AAAI
  Press.

\bibitem{Borgwardt+2018}
Stefan Borgwardt, {\.{I}}smail~{\.{I}}lkan Ceylan, and Thomas Lukasiewicz.
\newblock {Recent Advances in Querying Probabilisitc Knowledge Bases}.
\newblock In {\em Proceedings of the 27th International Joint Conference on
  Artificial Intelligence (IJCAI '18)}, pages 5420--5426. International Joint
  Conferences on Artificial Intelligence, 2018.

\bibitem{Boulos+2005}
Jihad Boulos, Nilesh Dalvi, Bhushan Mandhani, Shobhit Mathur, Chris R\'{e}, and
  Dan Suciu.
\newblock {MYSTIQ: A System for Finding more Answers by Using Probabilities}.
\newblock In {\em Proceedings of the 2005 ACM SIGMOD International Conference
  on Management of Data (SIGMOD '05)}, pages 891--893, New York, NY, USA, 2005.
  ACM.

\bibitem{Ceylan+2016}
\.{I}smail~\.{I}lkan Ceylan, Adnan Darwiche, and Guy Van~den Broeck.
\newblock {Open-World Probabilistic Databases}.
\newblock In {\em Proceedings of the Fifteenth International Conference on
  Principles of Knowledge Representation and Reasoning (KR '16)}, pages
  339--348, Palo Alto, CA, USA, 2016. AAAI Press.

\bibitem{Daley+2003}
Daryl~John Daley and David Vere-Jones.
\newblock {\em {An Introduction to the Theory of Point Processes, Volume I:
  Elementary Theory and Models}}.
\newblock Probability and its Applications. Springer, New York, NY, USA, 2nd
  edition, 2003.

\bibitem{DeRaedt+2016}
Luc De~Raedt, Kristian Kersting, Sriraam Natarajan, and David Poole.
\newblock {\em {Statistical Relational Artificial Intelligence: Logic,
  Probability, and Computation}}.
\newblock Synthesis Lectures on Artificial Intelligence and Machine Learning.
  Morgan \&{} Claypool, San Rafael, CA, USA, 2016.

\bibitem{Dong+2014}
Xin Dong, Evgeniy Gabrilovich, Geremy Heitz, Wilko Horn, Ni~Lao, Kevin Murphy,
  Thomas Strohmann, Shaohua Sun, and Wei Zhang.
\newblock {Knowledge Vault: A Web-scale Approach to Probabilistic Knowledge
  Fusion}.
\newblock In {\em Proceedings of the 20th ACM SIGKDD International Conference
  on Knowledge Discovery and Data Mining (KDD '14)}, pages 601--610, New York,
  NY, USA, 2014. ACM.

\bibitem{Fristedt+1997}
Bert~E. Fristedt and Lawrence~F. Gray.
\newblock {\em {A Modern Approach to Probabilitiy Theory}}.
\newblock Probability and its Applications. Birkh{\"a}user, Cambridge, MA, USA,
  1st edition, 1997.

\bibitem{Green2009}
Todd~J. Green.
\newblock {Models for Incomplete and Probabilistic Information}.
\newblock In Charu~C. Aggarwal, editor, {\em Managing and Mining Uncertain
  Data}, volume~35 of {\em Advances in Database Systems}, chapter~2, pages
  9--43. Springer, Boston, MA, USA, 2009.

\bibitem{Huang+2009}
Jiewen Huang, Lyublena Antova, Christoph Koch, and Dan Olteanu.
\newblock {MayBMS: A Probabilistic Database Management System}.
\newblock In {\em Proceedings of the 2009 ACM SIGMOD International Conference
  on Management of Data (SIGMOD '09)}, pages 1071--1074, New York, NY, USA,
  2009. ACM.

\bibitem{Imielinski+1984}
Tomasz Imieli\'{n}ski and Witold Lipski, Jr.
\newblock {Incomplete Information in Relational Databases}.
\newblock {\em Journal of the ACM (JACM)}, 31(4):761--701, 1984.

\bibitem{Jampani+2008}
Ravindranath Jampani, Fei Xu, Mingxi Wu, Luis~Leopoldo Perez,
  Christopher~Matthew Jermaine, and Peter~Jay Haas.
\newblock {MCDB: A Monte Carlo Approach to Managing Uncertain Data}.
\newblock In {\em Proceedings of the 2008 ACM SIGMOD International Conference
  on Management of Data (SIGMOD '08)}, pages 687--700, New York, NY, USA, 2008.
  ACM Press.

\bibitem{Kennedy+2010}
Oliver Kennedy and Christoph Koch.
\newblock {PIP: A Database System for Great and Small Expectations}.
\newblock In {\em Proceedings of the 26th International Conference on Data
  Engineering (ICDE '10)}, pages 157--168, Washington, DC, USA, 2010. IEEE.

\bibitem{Kimelfeld+2013}
Benny Kimelfeld and Pierre Senellart.
\newblock {Probabilistic XML: Models and Complexity}.
\newblock In {\em Advances in Probabilistic Databases for Uncertain Information
  Management}, volume 304 of {\em Studies in Fuzziness and Soft Computing},
  chapter~3, pages 39--66. Springer, Berlin and Heidelberg, Germany, 2013.

\bibitem{Knopp1996}
Konrad Knopp.
\newblock {\em {Theorie und Anwendung der unendlichen Reihen}}.
\newblock Springer, Berlin and Heidelberg, Germany, 6th edition, 1996.
\newblock An english translation of a previous edition is available under the
  title \emph{Theory and Application of Infinite Series}, 1990, published by
  Dover Publications, Mineola, NY, USA.

\bibitem{Libkin2018}
Leonid Libkin.
\newblock {Certain Answers Meet Zero-One Laws}.
\newblock In {\em Proceedings of the 37th ACM SIGMOD-SIGACT-SIGAI Symposium on
  Principles of Database Systems (PODS '18)}, pages 195--207, New York, NY,
  USA, 2018. ACM.

\bibitem{Lukasiewicz+2016}
Thomas Lukasiewicz and Dan Olteanu.
\newblock {Probabilistic Databases and Reasoning}.
\newblock
  \url{https://www.cs.ox.ac.uk/dan.olteanu/tutorials/olteanu-pdb-kr16.pdf},
  2016.
\newblock Tutorial at the Fifteenth International Conference on Principles of
  Knowledge Representation and Reasoning (KR '16).

\bibitem{Milch+2005}
Brian~Christopher Milch, Bhaskara Marthi, Stuart Russell, David Sontag,
  David~L. Ong, and Andrey Kolobov.
\newblock {BLOG: Probabilistic Models with Unknown Objects}.
\newblock In {\em Proceedings of the 19th International Joint Conference on
  Artificial Intelligence (IJCAI '05)}, pages 1352--1359, St. Louis, MO, USA,
  2005. Morgan Kaufmann.

\bibitem{Mitchell+2015}
Tom Mitchell, William Cohen, Estevam Hruschka, Partha Talukdar, Justin
  Betteridge, Andrew Carlson, Bhavana Dalvi~Mishra, Matthew Gardner, Bryan
  Kisiel, Jayant Krishnamurthy, Ni~Lao, Kathryn Mazaitis, Thahir Mohamed, Ndapa
  Nakashole, Emmanouil Platanios, Alan Ritter, Mehdi Samadi, Burr Settles,
  Richard Wang, Derry Wijaya, Abhinav Gupta, Xinlei Chen, Abulhair Saparov,
  Malcolm Greaves, and Joel Welling.
\newblock {Never-Ending Learning}.
\newblock In {\em Proceedings of the 29th AAAI Conference on Artificial
  Intelligence (AAAI '15)}, pages 2302--2310, Palo Alto, CA, USA, 2015. AAAI
  Press.

\bibitem{Niu+2012}
Feng Niu, Ce~Zhang, Christopher R{\'{e}}, and Jude~W. Shavlik.
\newblock {DeepDive: Web-scale Knowledge-Base Construction using Statistical
  Learning and Inference}.
\newblock In {\em Very Large Data Search: Proceedings of the Second
  International Workshop on Searching and Integrating New Web Data Sources
  (VLDS '12)}, pages 25--28, Aachen, Germany, 2012. CEUR Workshop Proceedings.

\bibitem{Reiter1978}
Raymond Reiter.
\newblock {On Closed World Data Bases}.
\newblock In Herve Gallaire and Jack Minker, editors, {\em Logic and Data
  Bases}, pages 55--76. Plenum Press, New York, NY, USA, 1st edition, 1978.

\bibitem{Renyi2007}
Alfred R\'enyi.
\newblock {\em {Foundations of Probability}}.
\newblock Dover Publications, Mineola, NY, USA, reprint edition, 2007.

\bibitem{Simmons2013}
David Simmons.
\newblock {Infinite Distributive Law}.
\newblock Mathematics Stack Exchange, 2013.
\newblock \url{https://math.stackexchange.com/q/509318}.

\bibitem{Singh+2008}
Sarvjeet Singh, Chris Mayfield, Rahul Shah, Sunil Prabhakar, Susanne Hambrusch,
  Jennifer Neville, and Reynold Cheng.
\newblock {Database Support for Probabilistic Attributes and Tuples}.
\newblock In {\em 2008 {IEEE} 24th International Conference on Data Engineering
  (ICDE '08)}, pages 1053--1061, Washington, DC, USA, 2008. IEEE Computer
  Society.

\bibitem{Singla+2007}
Parag Singla and Pedro Domingos.
\newblock Markov logic in infinite domains.
\newblock In {\em Proceedings of the Twenty-Third Conference on Uncertainty in
  Artificial Intelligence (UAI '07)}, pages 368--375, Arlington, VA, USA, 2007.
  AUAI Press.

\bibitem{Spencer2001}
Joel Spencer.
\newblock {\em {The Strange Logic of Random Graphs}}, volume~22 of {\em
  Algorithms and Combinatorics}.
\newblock Springer, Berlin and Heidelberg, Germany, 2001.

\bibitem{Suciu+2011}
Dan Suciu, Dan Olteanu, Christopher R\'{e}, and Christoph Koch.
\newblock {\em {Probabilistic Databases}}, volume~16 of {\em Synthesis Lectures
  on Data Management}.
\newblock Morgan \&{} Claypool, San Rafael, CA, USA, 1st edition, 2011.

\bibitem{VandenBroeck+2017}
Guy Van~den Broeck and Dan Suciu.
\newblock {Query Processing on Probabilistic Data: A Survey}.
\newblock {\em Foundations and Trends\textsuperscript{\textregistered{}} in
  Databases}, 7(3--4):197--341, 2017.

\bibitem{vanderMeyden1998}
Ron van~der Meyden.
\newblock {Logical Approaches to Incomplete Information: A Survey}.
\newblock In Jan Chomicki and Gunter Saake, editors, {\em Logics for Databases
  and Information Systems}, pages 307--356. Kluwer Academic Publishers,
  Norwell, MA, USA, 1998.

\end{thebibliography}
\end{document}